\documentclass{article}
\usepackage{amssymb} 
\usepackage{amsthm}   
\usepackage{epsfig}
\usepackage{natbib}
\usepackage{subfigure}
\usepackage{amsmath}
\usepackage{amsfonts}
\usepackage{natbib}
\usepackage{epsfig}
\usepackage{subfigure}
\usepackage{ifthen}
\newfont{\handw}{cmmi10 scaled 1200}

\newtheorem{Prop}{Proposition}[section]
\newtheorem{Lem}[Prop]{Lemma}
\newtheorem{Th}[Prop]{Theorem}

\newtheorem{Rm}[Prop]{Remark}
\newtheorem{Def}[Prop]{Definition}
\newtheorem{Ex}[Prop]{Example}
\newtheorem{Cor}[Prop]{Corollary}
 
\newtheorem{Cond}[Prop]{Condition}
\newfont{\smcal}{cmu10 scaled 1200}

\newcommand{\re}{\operatorname {Re}}

\newcommand{\diag}{\operatorname {diag}}

\newcommand{\E}{\operatorname {\mathbb E}}
\newcommand{\Prob}{\operatorname {\mathbb P}}

\newcommand{\tr}{\operatorname {tr}}
\newcommand{\rank}{\operatorname {rank}}

\newcommand{\argmin}{\operatorname {argmin}}

\begin{document}
\title{On the Meaning of Mean Shape}
\author{Stephan F. Huckemann}
  \date{}
        \maketitle

\begin{abstract}
	Various concepts of mean shape previously unrelated in the literature are brought into relation. In particular for non-manifolds such as Kendall's 3D shape space, this paper answers the question, for which means one may apply a two-sample test. The answer is positive if intrinsic or Ziezold means are used. The underlying general result of manifold stability of a mean on a shape space, the quotient due to an isometric action of a compact Lie group on a Riemannian manifold, blends the Slice Theorem from differential geometry with the statistics of shape. For 3D Procrustes means, however, a counterexample is given.  To further elucidate on subtleties of means, for spheres and Kendall's shape spaces, a first order relationship between intrinsic, residual/Procrustean and extrinsic/Ziezold means is derived stating that for high concentration the latter approximately divides the (generalized) geodesic segment between the former two by the ratio $1:3$. This fact, consequences of coordinate choices for the power of tests and other details, e.g. that extrinsic Schoenberg means may increase dimension are discussed and illustrated by simulations and exemplary datasets.
\end{abstract}
\par
\vspace{9pt}
\par
\vspace{9pt}Figure 
\noindent {\it AMS 2000 Subject Classification:} \begin{minipage}[t]{6cm}
Primary 60D05\\ Secondary 62H11
 \end{minipage}
\par

\section{Introduction}\label{intro-scn}

	The analysis of shape may be counted among the very early activities of mankind; be it for representation on cultural artefacts, or for morphological, biological and medical applications. In modern days shape analysis is gaining increased momentum in computer vision, image analysis, biomedicine and many other fields. For a recent overview cf. \cite{KY06}. 

	A \emph{shape space} can be viewed as the quotient of a Riemannian manifold -- e.g. the pre-shape sphere of centered unit size landmark configurations -- modulo the isometric and proper action of a Lie group (cf. \cite{Bre72}), conveying shape equivalence -- e.g. the group of rotations, cf. \citep[Chapter 11]{KBCL99}. Thus, it carries the canonical quotient structure of a union of manifold strata of different dimensions, which give in general a Riemannian \emph{manifold part} -- possibly with singularities comprising the \emph{non-manifold part} of \emph{non-regular shapes} at some of which sectional curvatures may tend to infinity, cf. \citep[Chapter 7.3]{KBCL99} as well as \cite{HHM07}. 

	In a Euclidean space, there is a clear and unique concept of a mean in terms of least squares minimization: the arithmetic average. Generalizing to manifolds, however, the concept of expectation, average or \emph{mean}  is surprisingly non trivial and not at all canonical. In fact, it resulted in an overwhelming number of different concepts of means, each defined by a specific concept of a distance, all of which are identical for the Euclidean distance in a Euclidean space. More precisely, with every embedding in a Euclidean space come specific \emph{extrinsic} and \emph{residual means} and with every Riemannian structure comes a specific \emph{intrinsic mean}. Furthermore, due to the non-Euclidean geometry, local minimizers introduced as \emph{Karcher means} by \cite{KWS90} may be different from global minimizers called \emph{Fr\'echet means} by \cite{Z77}, and, neither ones are necessarily unique. Nonetheless, carrying statistics over to manifolds, strong consistency (by \cite{Z77}, \cite{BP03}) and under suitable conditions, central limit theorems (CLTs) for such means have been derived (by \cite{J88}, \cite{HL96,HL98},  \cite{BP05} as well as \cite{H_Procrustes_10}). 
	On shape spaces, various other concepts of means  have been introduced, e.g. the famous \emph{Procrustes means} (cf. \cite{Z77},\cite{DM98}). As we show here, these means are related to the above ones via a 
	\emph{horizontal lifting} from the bottom quotient to the top manifold, cf. Table \ref{means:tab}. In particular, since there are many -- and often confusing -- variants of Procrustes means in the literatur this paper introduces the terminology of \emph{Procrustean means} standing for inheritance from residual means.
	\begin{table}[h!]
 	 	\centering
	 \fbox{$\begin{array}{c|c}
	  \mbox{manifold means} & \mbox{shape means}\\ \hline
	\mbox{intrinsic}&\mbox{intrinsic}\\
	\mbox{extrinsic}&\mbox{Ziezold}\\
	\mbox{residual} & \mbox{Procrustean}
	 \end{array}$}
	\caption{\it Three fundamental types of means on a shape space (right column) and their horizontal lifts to the respective manifold (left column). 
	\label{means:tab}} 
	\end{table}

	For a CLT to hold, a manifold structure locally relating to a Euclidean space is sufficient. 
	This leads to the question under which conditions it can be guaranteed that a mean shape lies on the manifold part.

	Due to strong consistency, 
	for a \emph{one-sample test} for a specific mean shape on the manifold part, it may be assumed that sample means eventually lie on the manifold part as well, thus making the above cited CLTs available. To date however, a two- and a multi-sample test could not be justified because of a lacking result on the following manifold stability. 
	\begin{Def} 
	A mean shape enjoys \emph{manifold stability} if it is assumed on the manifold part for any random shape assuming the manifold part with non-zero probability. 
	\end{Def}

	A key result of this paper establishes manifold stability for intrinsic and Ziezold means under the following condition.
	\begin{Cond}\label{realistic:cond}
	 On the non-manifold part the distribution of the random shape contains at most countably many point masses.
	\end{Cond}

	Since the non-manifold part is a null-set (e.g. \cite{Bre72}) under the projection of the Riemannian volume, this condition covers most realistic cases.

	We develop the corresponding theory for a general shape space quotient based on lifting a distribution on the shape space to the pre-shape space and subsequently exploiting the fact that intrinsic means are zeroes of an integral involving the Riemann exponential. The similar argument can be applied to Ziezold means but not to Procrustean means. More specifically, we develop the notion of a \emph{measurable horizontal lift} of the shape space except for its \emph{quotient cut locus} (introduced as well) to the pre-shape space. This requires the geometric concept of \emph{tubular neighborhoods admitting slices}.


	Curiously, the result applied to the finite dimensional subspaces exhausting the quotient shape space of closed planar curves with arbitrary initial point introduced by \cite{ZR72} and further studied by \cite{KSMJ04}, gives that the shape of the circle, since it is a singularity, can never be an intrinsic shape mean of non-circular curves.

	As a second curiosity, 3D full Procrustes means do not enjoy manifold stability in general, a counterexample involving low concentration is given. This is due to the fact that for low concentration, full Procrustes means may be `blinder' in comparison to intrinsic and Ziezold means to distributional changes far away from a mode.
	Included in this context is also a discussion of the \emph{Schoenberg means}, recently introduced by \cite{BandPat05} as well as by \cite{DKLW08} for the non-manifold Kendall reflection shape spaces, which in the ambient space, also allow for a CLT. Schoenberg means, as demonstrated, however, may feature `blindness' in comparison to intrinsic and Ziezold means, with respect to changes in the distribution of nearly degenerate shapes. 
	In a simulation we show that these features render Schoenberg means less effective for a discrimination involving degenerate or nearly degenerate shapes. 

	As a third curiosity, for spheres and Kendall's shape spaces, it is shown that, given uniqueness, with order of 
	concentration, 
	the (generalized) geodesic segment between the intrinsic mean and the residual/Procrustean mean 
	is in approximation divided by the extrinsic/Ziezold mean by the ratio $1:3$. This 
	first order relationship can be readily observed in existing data sets.
	In particular, this result supports the conjecture that Procrustean means of sufficiently concentrated distributions enjoy stability as well.

	This paper is structured as follows. 
	For convenience of the reader, first in Section \ref{Kendalls-ss:scn}, Kendall's shape spaces are introduced along with the specific result on manifold stability, followed by a classification of concepts of means on general shape spaces in Section \ref{Frechet-rho-means:scn}. 
	The rather technical Section \ref{convex:scn} develops horizontal lifting and establishes manifold stability, technical proofs are deferred to the appendix. 
	 In Section \ref{ext-means:scn} extrinsic Schoenberg means are discussed and 
	Section \ref{local:scn} tackles local effects of curvature on spheres and Kendall's shape spaces. 
%
	Section \ref{class_data_simulations_scn} illustrates practical consequences using classical data-sets as well as simulations. 
	Note that lacking stability 
	does not affect the validity of the Strong Law, on which the considerations on asymptotic distance in Sections \ref{local:scn} and \ref{class_data_simulations_scn} are based.

	An R-package for all of the computations performed is provided online: \cite{Hshapes}.

\section{Stability of Means on Kendall's Shape Spaces}\label{Kendalls-ss:scn}

	In the statistical analysis of similarity shapes based on landmark configurations, geometrical $m$-dimensional objects (usually $m=2,3$) are studied by placing $k>m$ \emph{landmarks} at specific locations of each object. Each object is then described by a matrix in the space $M(m,k)$ of $m\times k$ matrices, each of the $k$ columns denoting an $m$-dimensional landmark vector. $\langle x,y\rangle := \tr(xy^T)$ denotes the usual inner product with norm $\|x\| = \sqrt{\langle x,x\rangle}$. For convenience and without loss of generality for the considerations below, only \emph{centered} configurations are considered. Centering can be achieved by multiplying with a sub-Helmert matrix ${\cal H}\in M(k,k-1)$ 
	from the right, yielding a configuration $x{\cal H}$ in $M(m,k-1)$. For this and other centering methods cf. \citep[Chapter 2]{DM98}. 
	Excluding also all configurations with all landmarks coinciding 
	gives the space of \emph{configurations} 
	\begin{eqnarray*}
	F_m^k&:=& M(m,k-1) \setminus \{0\} \,.
	\end{eqnarray*}
	Since 
	only the similarity shape is of concern, we may assume that all configurations are contained in the unit sphere
	$ S_m^k :=\{x\in F_m^k: \|x\|=1\}$ called the \emph{pre-shape sphere}. Then, \emph{Kendall's shape space} is the canonical quotient 
	$$\Sigma_m^k := S_m^k/SO(m) = \{[x]:x\in S_m^k\}\mbox{ with the \emph{orbit} } [x] = \{gx:g\in SO(m)\}\,.$$
	In some applications reflections are also filtered out giving \emph{Kendall's reflection shape space}  $$R\Sigma_m^k := \Sigma_m^k/\{e,\widetilde{e}\}= S_m^k/O(m)\,.$$
	Here, $O(m) = \{g\in M(m,m): g^Tg=e\}$ denotes the orthogonal group with the unit matrix $e=\diag(1,\ldots,1)$,  $\widetilde{e} = \diag(-1,1,\ldots,1)$ 
	and $SO(m) = \{g\in O(m): \det(g)=1\} $ 
	is the special orthogonal group.
 
	For $1\leq j < m < k$ consider the isometric embedding
	\begin{eqnarray}\label{pre-shape-emb:eq} \begin{array}{rclcrcl} S_j^k &\hookrightarrow& S_m^k&:& x&\mapsto&\left(\begin{array}{c}x\\\hline 0\end{array}\right)\end{array}\,
	\end{eqnarray}
	giving rise to a canonical embedding $R\Sigma_j^k\hookrightarrow \Sigma_m^k$ which is isometric w.r.t. the canonical intrinsic distance, the Procrustean distance and the Ziezold distance, respectively, defined in Section \ref{Frechet-rho-means:scn}, cf. \citep[p. 29]{KBCL99}, cf. also Remark \ref{sharp:rm} below.

	We say that a configuration in $\mathbb R^m$ is \emph{$j$-dimensional}, or more precisely \emph{non-degenerate $j$-dimensional} if its preshape $x\in S_m^k$ is of rank $j$. Moreover, for $j\geq 3$ the shape spaces $\Sigma_j^k$ and $R\Sigma_j^k$ decompose into a \emph{manifold part} (defined in Section \ref{Frechet-rho-means:scn}, cf. also 
	Section \ref{ext-means:scn}) of \emph{regular shapes}
	$$(\Sigma_j^k)^* = \{[x]\in \Sigma_j^k: \rank(x) \geq j-1\}\mbox{ and }(R\Sigma_j^k)^* = \{[x]\in R\Sigma_j^k: \rank(x) =j\}\,,$$
	respectively, given by the shapes corresponding to configurations of at least dimension $j-1$ and $j$, respectively and a non void part of singular shapes corresponding to lower dimensional configurations, respectively.

	The following Theorem for intrinsic means, full Procrustes means and Ziezold means (also defined in Section \ref{Frechet-rho-means:scn}) of random elements taking values in $R\Sigma_j^k$ follows from Proposition \ref{res_means_sphere:prop}, Remark \ref{kend-mf-quot-mean:rm} and the fact that $R\Sigma_j^k\subset \Sigma_m^k$ contains all shapes in $\Sigma_m^k$ of configurations of dimension up to $j$, $1\leq j<m<k$.

	\begin{Th}\label{decr_dim:th}
	 Suppose that $X$ is a random shape on $\Sigma_m^k$ assumes shapes in $R\Sigma_j^k$  ($1\leq j < m<k$) with probability one. Then every full Procrustes mean shape of $X$ and every unique intrinsic or Ziezold mean shape under Condition \ref{realistic:cond} w.r.t. $(R\Sigma_j^k)^*$ corresponds to a configuration of dimension less than or equal to $j$.
	\end{Th}
 
	The following theorem is the application of the key result 
	applied to Kendall's shape spaces. 

	\begin{Th}[Stability Theorem for Intrinsic and Ziezold means]
	\label{Kendall_mean_dim} Let $X$ be a random shape on $\Sigma_m^k$, $0<m<k$, with unique intrinsic or Ziezold mean shape $[\mu] \in \Sigma_m^k$, $\mu \in S_m^k$ and let $1\leq j \leq m$ be the maximal dimension of configurations of shapes assumed by $X$ with non-zero probability. Suppose moreover that shapes of configurations of strictly lower dimensions are assumed with at most countably many point masses.
	\begin{enumerate}\item[(i)]
	If $j<m$ then $\mu$ corresponds to a non-degenerate $j$-dimensional configuration. 
	\item[(ii)]If $j=m$ then $\mu$ corresponds to a non-degenerate  configuration of dimension $m-1$ or $m$.
	\end{enumerate}
	\end{Th}

	\begin{proof}  
	Lemma \ref{cut_locus:lem} teaches that for Kendall's shape spaces, all quotient cut loci are void. Since 
	for Ziezold means, Remark \ref{kendall_book:rm} provides invariant optimal positioning and Remark \ref{half-sphere:rm} provides the validity of (\ref{Ziez-inj:cond}), Corollary 	\ref{pop_mean_reg:cor} applied to $R\Sigma_{j}^k$ as well as to $\Sigma_m^k$ states that intrinsic and Ziezold means are also assumed on the manifold parts of $R\Sigma_{j}^k$ and $\Sigma_m^k$, respectively. In conjunction with Theorem \ref{decr_dim:th}, this gives the assertion.	
	\end{proof}
 

	\begin{Rm}\label{sharp:rm} The result of Theorem \ref{Kendall_mean_dim} is sharp. 
	To see this, consider for  $\alpha >\beta >0$, $\alpha^2+\beta^2 =1$ the pre-shapes 
 	$$ x = \left(\begin{array}{ccc}\alpha &0&0\\0&\beta&0 \end{array}\right),~~y = \left(\begin{array}{ccc}\alpha &0&0\\0&-\beta&0 \end{array}\right)\mbox{ and }z = \left(\begin{array}{ccc}1 &0&0\\0&0&0 \end{array}\right)\, \in S_2^4\,.$$
	Then $x$ and $y$ correspond to non-degenerate two-dimensional quadrilateral configurations while $z$ corresponds to a one-dimensional (collinear) quadrilateral. Still, $[z]$ is regular in $\Sigma_2^4$ and it is the intrinsic and Ziezold mean of $[x]$ and $[y]$ in $\Sigma_2^4$. Under the embedding $R\Sigma_2^4\hookrightarrow \Sigma_3^4$ we have the pre-shapes
	$$x' = \left(\begin{array}{ccc}\alpha &0&0\\0&\beta&0 \\0&0&0\end{array}\right),~~y' = \left(\begin{array}{ccc}\alpha &0&0\\0&-\beta&0\\0&0&0 \end{array}\right)\mbox{ and }z' = \left(\begin{array}{ccc}1 &0&0\\0&0&0 \\0&0&0\end{array}\right)\in S_3^4\,.$$
	Just as $[x] = [y]$ in $R\Sigma_2^4$ so do $x'$ and $y'$ have regular and identical shape in $\Sigma_3^4$. However, $[z']$ is not regular and it is not the intrinsic or Ziezold mean in $\Sigma_3^4$. 
	\end{Rm}

\section{Fundamental Types of Means} 
\label{Frechet-rho-means:scn}

	In the previous section we introduced Kendall's \emph{shape} and \emph{reflection shape space} based on invariance  under similarity transformations and, including reflections, respectively. Invariance under congruence transformations only leads to Kendall's \emph{size-and-shape space}. More generally in image analysis, invariance may also be considered under the affine or projective group, cf. \cite{MP01,MP05}. A different yet also very popular popular set of shape spaces for two-dimensional configurations modulo the group of similarities has been introduced by \cite{ZR72}. Instead of building on a finite dimensional Euclidean matrix space modeling landmarks, the basic ingredient of these spaces modeling closed planar unit speed curves is the infinite dimensional Hilbert space of Fourier series, cf. \cite{KSMJ04}. In practice for numerical computations, only finitely many Fourier coefficients are considered.

	To start with, a shape space is a metric space $(Q,d)$.
	For this entire paper suppose that $X, X_1,X_2,\ldots$ are i.i.d. random elements mapping from an abstract probability space $(\Omega,\cal A,\Prob)$ to $(Q,d)$ equipped with its self understood Borel $\sigma$-field. Here and in the following, \emph{measurable} will refer to the corresponding Borel $\sigma$-algebras, respectively. Moreover, denote by $\mathbb E(Y)$ the classical expected value of a random vector $Y$ on a $D$-dimensional Euclidean space $\mathbb R^D$, if existent. 

	\begin{Def}\label{Frechet_means:def} For a continuous function $\rho:Q\times Q \to [0,\infty)$ define the \emph{set of  population Fr\'echet $\rho$-means} by
	$$ E^{(\rho)}(X) = \argmin_{\mu\in Q} \mathbb E\big(\rho(X,\mu)^2\big) 
	\,.$$
	For $\omega\in \Omega$ denote the \emph{set of sample Fr\'echet $\rho$-means} by
	$$ E^{(\rho)}_n(\omega) = \argmin_{\mu\in Q} \sum_{j=1}^n \rho\big(X_j(\omega),\mu\big)^2\,.$$
	\end{Def}

	 	By continuity of $\rho$, the $\rho$-means are closed sets, additionally, sample $\rho$-means are random sets, all of which may be empty. For our purpose here, we rely on the definition of \emph{random closed sets} as introduced and studied by \cite{Choq54}, \cite{Kend74} and \cite{Math75}. Since their original definition for $\rho=d$ by \cite{F48} such means have found much interest. 

	\paragraph{Intrinsic means.} Independently, for a connected Riemannian manifold with geodesic distance $\rho^{(i)}$, \cite{KN69} defined the corresponding means as \emph{centers of gravity}. They are nowadays also well known as \emph{intrinsic means} by \cite{BP03,BP05}. 

	\paragraph{Extrinsic means.} W.r.t. the chordal or \emph{extrinsic metric} $\rho^{(e)}$ due to an embedding of a Riemannian manifold in an ambient Euclidean space, Fr\'echet $\rho$-means have been called \emph{mean locations} by \cite{HL96} or \emph{extrinsic means}  by \cite{BP03}. 

	More precisely, let $Q=M\subset \mathbb R^D$ be a complete Riemannian  manifold embedded in a Euclidean space $\mathbb R^D$ with standard inner product $\langle x,y\rangle$, $\|x\|=\sqrt{\langle x,x\rangle}, \rho^{(e)}(x,y)=\|x-y\|$ and let $\Phi: \mathbb R^D \to M$ denote the orthogonal projection, $\Phi(x) = {\rm argmin}_{p\in M}\|x-p\|$. For any Riemannian manifold an embedding that is even isometric can be found for $D$ sufficiently large, see \cite{Na56}. Due to an extension of Sard's Theorem by \citep[p.12]{BP03} for a closed manifold, $\Phi$ is univalent up to a set of Lebesgue measure zero. Then the set of extrinsic means is given by the set of images $\Phi\big(\mathbb E(Y)\big)$ where $Y$ denotes $X$ viewed as taking values in $\mathbb R^D$ (cf.  \cite{BP03}).

	\paragraph{Residual means.} In this context, setting $\rho^{(r)}(p,{p'}) = \|d\Phi_{p'}(p-{p'})\|$ ($p,{p'}\in M$) with the derivative $d\Phi_{p'}$ at ${p'}$ yielding the orthogonal projection to the embedded tangent space $T_{p'}\mathbb R^D \to T_{p'}M\subset T_{p'}\mathbb R^D$, call the corresponding mean sets $E^{(\rho^{(r)})}(X)$ and $E^{(\rho^{(r)})}_n(\omega)$, the sets of \emph{residual population means} and  \emph{residual sample means}, respectively. For two-spheres, $\rho^{(r)}(p,{p'})$ has been studied under the name of \emph{crude residuals} by \cite{J88}. On unit-spheres
	\begin{eqnarray}\label{res_dist_sphere:def}
	 \rho^{(r)}(p,{p'}) =\|p - \langle p,{p'}\rangle {p'}\| = \sqrt{1-\langle p,{p'}\rangle^2} = \rho^{(r)}({p'},p)
	\end{eqnarray}
	is a quasi-metric (symmetric, vanishing on the diagonal $p={p'}$ and satisfying the triangle inequality). On general manifolds, however, the \emph{residual distance} $\rho^{(r)}$ may be neither symmetric nor satisfying the triangle inequality. 
	~\\

	Obviously, for $X$ uniformly distributed on a sphere, the entire sphere is identical with the set of intrinsic, extrinsic and residual means: non-unique intrinsic and extrinsic means may depend counterintuitively on the dimension of the ambient space. 
	Here is a simple illustration.

	\begin{Prop}\label{spherical_means:thm} Suppose that $X$ is a random point on a unit sphere $S^{D-1}$ that is uniformly distributed on a unit subsphere $S$. Then 
	 	\begin{enumerate}
	 	 \item[(i)] every point on $S^{D-1}$ is an extrinsic mean and,
	 	 \item[(ii)] if $S$ is a proper subsphere then the set of intrinsic means is equal to the unit subsphere $S'$ orthogonal to $S$.
	 	\end{enumerate}
	\end{Prop}

	\begin{proof} The first assertion is a consequence of
	$ \rho^{(e)}(x,y)^2 + \rho^{(e)}(x,-y)^2 = 4$ for every $x,y\in S^{D-1}$.
	The second assertion follows from
	$$ \rho^{(i)}(x,y)^2 + \rho^{(i)}(x,-y)^2 = \rho^{(i)}(x,y)^2 + \big(\pi -\rho^{(i)}(x,y)\big)^2~\geq~ \frac{\pi^2}{2} \mbox{ for every }x,y\in S^{D-1}\,$$
	for the intrinsic distance $\rho^{(i)}(x,y) = 2\arcsin(\|x-y\|/2)$ with equality if and only if $x$ is orthogonal to $y$.
	\end{proof}


	\begin{Prop}\label{res_means_sphere:prop} If a random point $X$ on a unit sphere is a.s. contained in a unit subsphere $S$ then $S$ contains every residual mean as well as every unique intrinsic or extrinsic mean .\end{Prop}

	\begin{proof} Suppose that $x = v+\nu$ is a mean of $X$ with $v/\|v\| \in S$ and $\nu\in S^{D-1}$ normal to $S$. Since $1-\langle X,v+\nu\rangle^2 = 1-\langle X,v\rangle^2 \geq 1 - \langle X,v\rangle^2/\|v\|^2$ a.s. with equality if and only if $\nu=0$, the assertion for residual means follows at once from (\ref{res_dist_sphere:def}). For intrinsic and extrinsic means we argue with
	$\|X-(v+\nu)\| = \|X-(v-\nu)\|$ a.s. yielding $\nu=0$ in case of uniqueness.
	\end{proof}
%
%

	Let us now incorporate more of the structure common to shape spaces. The following definition is due to \citep[p. 249]{KBCL99}. We additionally require that the group acting be compact in order to ensure that the quotient be Hausdorff. More generally, one could assume a non-compact group acting \emph{properly}, cf. \cite{P61}.

	\begin{Def}\label{shape:def}
	A complete connected finite-dimensional Riemannian manifold $M$ with geodesic distance $d_M$ on which a compact Lie group $G$ acts isometrically from the left is called a \emph{pre-shape space}. Moreover the canonical quotient 
	$$\pi: M \to Q:= M/G = \{[p]: p \in M\}\mbox{  with the \emph{orbit} } [p] = \{gp: g\in G\}\,,$$ 
	is called a \emph{shape space}.
	\end{Def}

	 As a consequence of the isometric action we have that $d_M(gp,{p'}) = d(p,g^{-1}{p'})$ for all $p,{p'}\in M$, $g\in G$. For $p,{p'}\in M$ we say that $p$ is in \emph{optimal position} to ${p'}$ if $d_M(p,{p'}) =\min_{g\in G}d_M(gp,{p'})$, the minimum is attained since $G$ is compact. As is well known (e.g. \citep[p. 179]{Bre72}) there is an open and dense submanifold $M^*$ of $M$ such that the canonical quotient $Q^* = M^*/G$ restricted to $M^*$ carries a natural manifold structure also being open and dense in $Q$. Elements in $M^*$ and $Q^*$, respectively, are called \emph{regular}, the complementary elements are \emph{singular}; $Q^*$ is the \emph{manifold part} of $Q$. 

	\paragraph{Intrinsic means on shape spaces.}
	The canonical quotient distance 
	$$d_Q([p],[{p'}]) :=\min_{g\in G}d_M(gp,{p'}) = \min_{g,h\in G}d_M(gp,h{p'})$$
	is called \emph{intrinsic distance} and the corresponding $d_Q$-Fr\'echet mean sets are called \emph{intrinsic means}. Note that the intrinsic distance on $Q^*$ is equal to the canonical geodesic distance. 
	
	\paragraph{Ziezold and Procrustean means.}
	Now, assume that we have an embedding with orthogonal projection $\Phi: \mathbb R^D \to M\subset\mathbb R^D$ as above. If the action of $G$ is isometric w.r.t. the extrinsic metric, i.e. if $\|gp-g{p'}\| = \|p-{p'}\|$ for all $p,{p'}\in M$ and $g\in G$ then call 
	\begin{eqnarray*}
	\rho^{(z)}_Q([p],[{p'}])&:=&\min_{g\in G}\|gp-{p'}\|\mbox{~~~and~~~}\\ 
	\rho^{(p)}_Q([p],[{p'}])&:=&\min_{\footnotesize\begin{array}{l}g\in G,~ gp \mbox{ in}\\ \mbox{opt. pos. to }{p'}\end{array}}\|d\Phi_{p'}(gp-{p'})\|	 	
	\end{eqnarray*}
 	the \emph{Ziezold distance} and the \emph{Procrustean distance} on $Q$, respectively. Call the corresponding population and sample Fr\'echet $\rho^{(z)}_Q$-means, respectively, the sets of population and sample  \emph{Ziezold means}, respectively. Similarly, call the corresponding population and sample Fr\'echet $\rho^{(p)}_Q$-means, respectively, the sets of population and sample  \emph{Procrustean means}, respectively.

	We say that \emph{optimal positioning is 
	invariant} if for all $p,{p'} \in M$ and $g^*\in G$,
	$$d_M(g^*p,{p'}) = \min_{g\in G}d_M(gp,{p'}) \Leftrightarrow \|g^*p-{p'}\| = \min_{g\in G}\|gp-{p'}\|\,.$$

	\begin{Rm}\label{kendall_book:rm}
	Indeed for $Q=\Sigma_m^k, R\Sigma_m^k$,  optimal positioning is invariant (cf. \citep[p. 206]{KBCL99}), Procrustean means coincide with means of \emph{general Procrustes analysis} introduced by \cite{Gow} and Ziezold means coincide with means as introduced by \cite{Z94} for $\Sigma_2^k$. Moreover for $\Sigma_2^k$, Procrustean means agree with extrinsic means w.r.t. the Veronese-Whitney embedding, cf. \cite{BP03} and Section \ref{ext-means:scn}. 
	\end{Rm}

	Procrustean means on $\Sigma_m^k$ are also called \emph{full Procrustes means} in the literature to distinguish them from \emph{partial Procrustes means} on the \emph{size-and-shape spaces} not further discussed here (e.g. \cite{DM98}). We only note that partial Procrustes means are identical to the respective intrinsic, Procrustean and Ziezold means which on size-and-shape spaces, all agree with one another.

\section{Horizontal Lifting and Manifold Stability}\label{convex:scn}
	In this section we derive a measurable horizontal lifting and the stability theorem underlying Theorem \ref{Kendall_mean_dim}. To this end we first recall how a shape space is made up from manifold strata of varying dimensions. Unless otherwise referenced, we use basic terminology that can be found in any standard textbook on differential geometry, e.g. \cite{KN63,KN69}. For the results derived here we assume that the shape space is a quotient modulo a compact group. We note that 
	these results remain valid in the more general case of a non-compact group acting \emph{properly}, cf. \cite{P61}. 

\subsection{Preliminaries}\label{prelims:scn}
	Assume that $Q=M/G$ is a shape space as in Definition \ref{shape:def}. 
	$T_pM$ is the tangent space of $M$ at $p\in M$ and $\exp_p$ denotes the \emph{Riemannian exponential} at $p$. Recall that on a Riemannian manifold the \emph{cut locus} $C(p)$ of $p$ comprises all points $q$ such that the extension of a length  minimizing geodesic joining $p$ with $q$ is no longer minimizing beyond $q$. In consequence, on a complete and connected manifold $M$ we have for every $p'\in M$ that there is $v' \in T_pM$ such that $p'= \exp_p v'$ while $v' = \exp^{-1}_pp'$ of minimal modulus is uniquely determined as long as $p'\in M \setminus C(p)$.  It is well known that the cut locus has measure zero in the sense that its image in any local chart has Lebesgue measure zero. From now on we call the cut locus the \emph{manifold cut locus} in order to distinguish it from the \emph{quotient cut locus} $C^{quot}(q)$ of $q\in Q$ which we define as
	$C^{quot}(q) := \{[p']: p'\in C(p) \mbox{ is in optimal position to some } p\in q\}\,.$
	Due to the isometric action we have for any $p\in q$ that 
	\begin{eqnarray}\label{quot-cut-loc:eq}
	C^{quot}(q) &=& \{[p']: p'\in C(p) \mbox{ is in optimal position to } p\}\subset \pi\big(C(p)\big). 	 
	\end{eqnarray}
	The following Lemma teaches that in general, the projection of the manifold cut locus, the manifold cut locus of the manifold part $Q^*$ and the quotient cut locus are different. In particular, quotient cut loci are void in the special case of Kendall's shape spaces.

	\begin{Lem}\label{cut_locus:lem}$C(q)\neq \emptyset$ for every $q\in \Sigma_2^k$ while $C^{quot}(q)=\emptyset\mbox{ for all }q\in \Sigma_m^k\,.$ Similarly $C^{quot}(q)=\emptyset\mbox{ for all }q\in R\Sigma_m^k\,.$\end{Lem}

	\begin{proof} The first assertion follows from the fact that $\Sigma_2^k$ is a compact manifold. For the second assertion consider $[p]\in \Sigma_m^k$. Since $C(p) = \{-p\}$ for $p\in S_m^k$ and $[p]=[-p]$ for even $m$ as well as for odd $m$ if $p$ is not regular, and, since $p,-p$ are not in optimal position for odd $m$ if $p$ is regular, we have that $C^{quot}([p])=\emptyset$. The third assertion follows from the fact that $[p]=[-p]$ for all $[p] \in R\Sigma_m^k$.
	\end{proof}

	Next we collect consequences of the isometric Lie group action, see \cite{Bre72}. 

	\begin{enumerate}
	 \item[(A)] With the \emph{isotropy group} $I_p=\{g\in G: gp =p\}$ for $p\in M$, every orbit carries the natural structure of a coset space $[p] \cong G/I_p$. Moreover, $p' \in M$ is of \emph{orbit type} $(G/I_p)$ if $I_{p'} = gI_pg^{-1}=I_{gp}$ for a suitable $g\in G$. If $I_p \subset I_{gp'}$ for suitable $g\in G$ then $p'$ is of \emph{lower orbit type} than $p$ and $p$ is of \emph{higher orbit type} than $p'$.
	\item[(B)] The pre-shapes of equal orbit type 
	$M^{(I_p)} := \{p'\in M:\mbox{ $p'$ is of orbit type }\linebreak  (G/I_p)\}$
	and the corresponding shapes $Q^{(I_p)}:=\{[p']:p'\in M^{(I_p)}\}$ are manifolds in $M$ and $Q$, respectively. Moreover, for $q\in Q$ denote by $Q^{(q)}$ the shapes of higher orbit type.

	\item[(C)]The orthogonal complement $H_pM$  in $T_pM$ of the tangent space $T_p[p]$ along the orbit is called the \emph{horizontal space}: $T_pM = T_p[p]\oplus H_pM$.
	\item[(D)] The \emph{Slice Theorem} 
	states that every $p\in M$ has a \emph{tubular neighborhood} $[p]\subset U \subset M$ such that with a suitable subset $D\subset H_pM$ the \emph{twisted product} $\exp_p D\times_{I_p} G$ is diffeomorphic with $U$. Here, the twisted product is the natural topological quotient of the product space $\exp_p D\times G$ modulo the equivalence 
	$$ (\exp_pv,g) \sim_{I_p} (\exp_pv', g')~\Leftrightarrow \exists h \in I_p\mbox{ such that } v'=dh v,~g'=gh^{-1}\,.$$
	We then say that the tubular neighborhood $U$ \emph{admits a slice} $\exp_p D$ \emph{via} 
	$U\cong \exp_p D\times_{I_p}G$.
	\item[(E)] Every $p\in M$ has a tubular neighborhood $U$ of $p$ that admits a slice $\exp_p D$ such
	that every $p' \in \exp_pD$ is in optimal position to $p$. Moreover, for any tubular neighborhood $U$ admitting a slice  $\exp_p D$, all points $p'\in U$ are of orbit type higher than or equal to the orbit type of $p$ and only finitely many orbit types occur in $U$. If $p$ is regular, i.e. of maximal orbit type, then the product is trivial: $\exp_p D \times_{I_p} G \cong \exp_p D \times G/I_p$. 
	\end{enumerate}

	Finally let us extend the following uniqueness property for the intrinsic distance to the Ziezold distance.
	The differential of the mapping $f^{p'}_{int} : M\setminus C(p') \to [0,\infty)$ defined by $f^{p'}_{int}(p) = d_M(p,\exp_pp')^2$ is given by $d f^{p'}_{int}(p) = -2v$ with $v=\exp_p^{-1}p'$ (cf. \citep[p. 110]{KN69}, \cite{Ka77}). Hence, we have for $p_1,p_2 \in M\setminus C(p)$ that 
	\begin{eqnarray}\label{intr-inj:cond}
	d f^{p_1}_{int}(p)&=&  df^{p_2}_{int}(p) ~\Leftrightarrow~ p_1=p_2\,.
	\end{eqnarray}
	In view of the extrinsic distance let $f^{p'}_{ext} : M\setminus C(p')  \to [0,\infty)$ be defined by $f^{p'}_{ext}(p)=\|p-p'\|^2=\|p-\exp_p(\exp_p^{-1}p')\|^2$. Mimicking (\ref{intr-inj:cond}) introduce the following condition
	\begin{eqnarray}\label{Ziez-inj:cond}
	 d f^{p_1}_{ext}(p) = d f^{p_2}_{ext}(p) &\Leftrightarrow& p_1=p_2
	\end{eqnarray}
	 for $p_1,p_2 \in M\setminus C(p)$. 
	\begin{Rm}\label{half-sphere:rm}
	 (\ref{Ziez-inj:cond}) is valid on closed half spheres since on the unit sphere
	$$ df_{ext}^{p'}(p) = -2\,\frac{v}{\|v\|}\,\sin(\|v\|)\mbox{ with } v=\exp_p^{-1}p'\,.$$
	\end{Rm}

\subsection{A Measurable Horizontal Lift}

	In order to establish the stability of means in Theorem \ref{population_mean_iso_grp:thm} in the following Section \ref{population_convex:scn}, here we lift a random shape $X$ from $Q$ horizontally to a random pre-shape $Y$ on $M$. In order to do so we need to guarantee the measurability of the horizontal lift in Theorem \ref{glob_hor_measurable_lift:thm} below, the proof of which can be found in the appendix. 

	Before continuing, let us consider a simple example for illustration. Suppose that $G=S^1\subset \mathbb C$ acts  on $M=\mathbb C$ by complex scalar  multiplication. Then $[0,\infty) \cong Q=M/G$ having the two orbit types $(S^1/I_0) = \{1\}$ and $(S^1/I_1) = S^1$ gives rise to $Q^{(I_0)} = \{0\}$ and $Q^{(I_1)} = (0,\infty)$. Obviously, $M$ admits a global slice via the polar decomposition $[0,\infty) \times_{S^1} S^1 =\{0\} \cup \big((0,\infty)\times S^1\big)\cong M$ about $0\in M$ (the Riemannian exponential is the identity if $T_0\mathbb C$ is identified with $\mathbb C$). Here, $X$ can be identified with its horizontal lift $Y$ to the global slice $[0,\infty) \subset M$. If, say, $X$ is uniformly distributed on $[1,2]$ then $\mathbb P\{X\in Q^{(I_1)} \}>0$. In this case the stability theorem states the obvious fact that $0 \in  Q^{(I_0)}$ cannot be a mean of $X$. 

	\begin{Def}\label{hor-lift:def}
	Call a measurable subset $L \subset M$ a \emph{measurable horizontal lift} of a measurable subset $R$ of $M/G$ \emph{in optimal position to $p \in M$} if
	\begin{enumerate}
	 \item the canonical projection $L \to R\subset M/G$ is surjective,
	\item every $p' \in L$ is in optimal position to $p\in L$,
	\item every orbit $[p']$ of $p'\in L$ meets $L$ once.
	\end{enumerate}
	\end{Def}

	\begin{Th}\label{glob_hor_measurable_lift:thm}  Let $p\in [p]\in Q$ and $A\subset Q$ countable. Then there is a measurable horizontal lift $L$ of $Q^{([p])} \cup A$ in optimal position to $p$.
	\end{Th}

	\begin{Th}\label{pop_int_mean_bottom_int_mean_top:Th} Assume that $X$ is a random shape on $Q$ and 
	that there are $p\in M$  and $A\subset Q$ countable
	such that $X$ is supported by $\big(Q^{([p])} \cup A\big)\setminus C^{quot}([p])$. With a measurable horizontal lift $L$ of $\big(Q^{([p])}\cup A\big)\setminus C^{quot}([p])$ in optimal position to $p$ define the random element $Y$ on $L\subset M$ by $\pi\circ Y = X$. 
	\begin{enumerate}\item[(i)]
	If $[p]$ is an intrinsic mean of $X$ on $Q$, then $p$ is an intrinsic mean of $Y$ on $M$ and
	$$\mathbb E(\exp^{-1}_pY)=0\,.$$
	\item[(ii)]
	If $[p]$ is a Ziezold mean of $X$ on $Q$ and optimal positioning is invariant, then $p$ is an extrinsic mean of $Y$ on $M$ and
	$$\mathbb E\big( d f_{ext}^{Y}(p)\big)=0\,.$$
	\item[(iii)] If $[p]$ is a Procrustean mean of $X$ on $Q$ and optimal positioning is invariant, then $p$ is a residual mean of $Y$ on $M$.
	\end{enumerate}
	\end{Th} 

	\begin{proof} 
	Suppose that $[p]$ is an intrinsic mean of $X$. If $p$ would not be an intrinsic mean of $Y$, there would  some $M\ni p'\neq p$  leading to the contradiction
	\begin{eqnarray*}
	  \mathbb E \big(d_Q([p'], X)^2\big) &=&  \mathbb E \big(d_M(p',Y)^2\big)
	~<~ \mathbb E \big(d_M(p,Y)^2\big) ~=~  \mathbb E \big(d_Q([p], X)^2\big)\,.
	\end{eqnarray*}
	Hence, $p$ is an intrinsic mean of $Y$. Replacing $d_Q$ by $\rho_Q^{(z)}$ and $d_M$ by the Euclidean distance gives the assertion for Ziezold and extrinsic means, respectively; and, using the Procrustean distance $\rho_Q^{(p)}$ on $Q$ as well as the residual distance $\rho^{(r)}_M$ on $M$ gives the assertion for Procrustean and residual means, respectively.

	For intrinsic means $p\in M$, the necessary condition $\mathbb E \big(\exp_{p}^{-1}Y\big) =0\,$ is developed in \citep[p. 110]{KN69}, cf. also \citep{Ka77} and \citep[p. 395]{KWS90}
	which yields the asserted equality in (i). By definition, the analog condition for an extrinsic mean $p\in M$ is
	$\mathbb E\big(d f^{Y}_{ext}(p)\big) =0$ 
	which is the asserted equality in (ii) completing the proof. 
	\end{proof}


	\begin{Rm}\label{Ziez-half-sphere:rm} Since the maximal intrinsic distance on $\Sigma_m^k$ and $R\Sigma_m^k$ is 
	$$\frac{\pi}{2} = \max_{x,y\in S_m^k}\mathop{\min}_{g\in SO(m)}\arccos\big(\tr(gxy^T)\big)=\max_{x,y\in S_m^k}\mathop{\min}_{g\in O(m)}\arccos\big(\tr(gxy^T)\big)\,,$$ 
	taking into account Remark \ref{half-sphere:rm}, condition (\ref{Ziez-inj:cond}) is satisfied for any horizontal lift in optimal position.	 
	\end{Rm}

\subsection{Manifold Stability}
\label{population_convex:scn}

	The proof of the following central theorem is deferred to the appendix.

	\begin{Th}\label{population_mean_iso_grp:thm}  Assume that $X$ is a random shape on $Q$, 
	  $p\in M$ and that $A\subset Q$ is countable such that $X$ is supported by $\big(Q^{([p])}\cup A\big)\setminus C^{quot}([p])$ 
	and let $p'\in [p'] \in Q^{([p])}$. If ${\Prob} \{X\in Q^{(I_{p'})}\} \neq 0$ and if either $[p]$ is
	\begin{enumerate}
	 \item[(i)]  an intrinsic mean of $X$ or
	 \item[(ii)] a Ziezold mean of $X$ while optimal positioning is invariant
	and (\ref{Ziez-inj:cond}) is valid,
	\end{enumerate}
	then $p'$ is of lower orbit type than $p$.
	\end{Th}

	We have at once the following Corollary.

	\begin{Cor}[Manifold Stability Theorem]\label{pop_mean_reg:cor}
	 Suppose that $X$ is a random shape on $Q$ that is supported by $Q\setminus C^{quot}([p])$ for some $[p] \in Q$ assuming the manifold part $Q^*$ with non-zero probability and having at most countably many point masses on the singular part. Then $[p]$ is regular if it 
	is an intrinsic mean of $X$, or if it 
	is a Ziezold mean, optimal positioning is invariant
	and (\ref{Ziez-inj:cond}) is valid. 
	\end{Cor}

	Since $Q\setminus Q^{(q)}$  is a null set in $Q$ for every $q\in Q$ (cf. \citep[p. 184]{Bre72}) and so is  $C^{quot}(q)$ -- by (\ref{quot-cut-loc:eq}) it is contained in the projection of a null set -- we have the following practical application.

	\begin{Cor}\label{cont_distr:cor}  Suppose that a random shape on $Q$ is absolutely continuously distributed w.r.t. the projection of the Riemannian volume on $M$. Then intrinsic and Ziezold population means are regular; the latter if optimal positioning is invariant and (\ref{Ziez-inj:cond}) is valid. And, intrinsic and Ziezold sample means are a.s. regular. 
	\end{Cor}

\subsection[Non-Stability for Procrustean Means]{An Example for Non-Stability of Procrustean Means}\label{Procrustes-convex:scn}

	Consider a random configuration $Z\in F_3^4$ assuming the collinear quadrangle $q_1$ with probability $2/3$ and the planar quadrangle $q_2$ with probability $1/3$ where
		$$ q_1=\left(\begin{array}{cccc}1&-1&0&0\\0&0&0&0\\0&0&0&0\end{array}\right),~~ q_2= \left(\begin{array}{cccc}1&1&-2&0\\\frac{1}{\sqrt{2}}&\frac{1}{\sqrt{2}}&\frac{1}{\sqrt{2}}&-\,\frac{3}{\sqrt{2}}\\0&0&0&0\end{array}\right)\,.$$
	Corresponding pre-shapes in optimal position w.r.t. the action of $SO(3)$ and $O(3)$ are given by 
		$$ p_1=\left(\begin{array}{ccc}1&0&0\\0&0&0\\0&0&0\end{array}\right),~~ p_2=\frac{1}{\sqrt{2}}\, \left(\begin{array}{ccc}0&1&0\\0&0&1\\0&0&0\end{array}\right)\,.$$
	Note that $[p_2]$ has regular shape in $(\Sigma_3^4)^*$. The full Procrustes mean of $[Z]\in \Sigma_3^4$ is easily computed to have the singular shape $[p_1]\in \Sigma_3^4\setminus (\Sigma_3^4)^*$, see Figure \ref{different_means_fig} as well as Examples \ref{different_means_ex} and Section \ref{blindness:scn}. Cf. also Remark \ref{sharp:rm}. 



\section{Extrinsic Means for Kendall's (Reflection) Shape Spaces}\label{ext-means:scn}

	Let us recall the well known \emph{Veronese-Whitney} embedding for Kendall's planar shape spaces $\Sigma_2^k$. Identify $F_2^k$ with $\mathbb C^{k-1}\setminus \{0\}$ such that every landmark column corresponds to a complex number. This means in particular that $z\in \mathbb C^{k-1}$ is a complex row(!)-vector. With the Hermitian conjugate $a^* = (\overline{a_{kj}})$ of a complex matrix $a=(a_{jk})$ the pre-shape sphere $S_2^k$ is identified with $\{z\in \mathbb C^{k-1}: zz^*=1\}$ on which $SO(2)$ identified with $S^1=\{\lambda \in\mathbb C: |\lambda|=1\}$ acts by complex scalar multiplication. Then the well known Hopf-Fibration mapping to complex projective space gives $\Sigma_2^k=S_2^k/S^1=\mathbb CP^{k-2}$. Moreover, denoting with $M(k-1,k-1,\mathbb C)$ all complex $(k-1)\times  (k-1)$ matrices, the Veronese-Whitney embedding is given by
	\begin{eqnarray*}
	 S_2^k/S^1 &\to& \{a \in M(k-1,k-1,\mathbb C): a^*=a\}\,,~~
	~[z] ~\mapsto~ z^*z\,.
	\end{eqnarray*}

	\begin{Rm}\label{VW-emb-iso:rm}
	The Procrustean  metric of $\Sigma_2^k$ is isometric with the Euclidean metric of $M(k-1,k-1,\mathbb C)$ since
	 we have $\langle z,w \rangle = \re(zw^*)$  for $z,w\in S_2^k$ and hence, $d^{(p)}_{\Sigma_2^k}([z],[w]) = \sqrt{1-wz^*zw^*} = \|w^*w-z^*z\|/\sqrt{2}$. 
	\end{Rm}

	The idea of the Veronese-Whitney embedding can be carried to the general case of shapes of arbitrary dimension $m\geq 2$. Even though the embedding given below is apt only for reflection shape space it can be applied to practical situations in similarity shape analysis whenever the geometrical objects considered have a common orientation. 
	As above, the number $k$ of landmarks is essential and will be considered fixed throughout this section; the dimension $1\leq m < k$, however, is lost in the embedding and needs to be retrieved by projection. To this end recall the embedding of $S_j^k$ in $S_m^k$ $(1\leq j \leq m)$ in (\ref{pre-shape-emb:eq}) which gives rise to a canonical embedding of $R\Sigma_j^m$ in $R\Sigma_m^k$.
	Moreover, consider the strata 
	$$(R\Sigma_m^k)^j := \{[x]\in R\Sigma_m^k: \rank(x)=j\},~~(\Sigma_m^k)^{j}:=\{[x] \in \Sigma_m^k: \rank(x)=j\}$$ 
	for $j=1,\ldots,m$, each of which carries a canonical manifold structure;  due to the above embedding, $(R\Sigma_m^k)^j$ will be identified with $(R\Sigma_j^k)^j$ such that
	$$ R\Sigma_m^k ~=~ \bigcup_{j=1}^m (R\Sigma_j^k)^j\,,$$
	and  $(R\Sigma_m^k)^j$ with $(\Sigma_m^k)^j$ in case of $j<m$.
	At this point we note that $SO(m)$ is connected, while $O(m)$ is not; and the consequences for the respective manifold parts, i.e. points of maximal orbit type:
	\begin{eqnarray}\label{mf-part:eq}
	 (\Sigma_m^k)^* &=& (\Sigma_m^k)^{m-1}\cup (\Sigma_m^k)^m\,,\quad
	 (R\Sigma_m^k)^* ~=~ (R\Sigma_m^k)^{m}\,.	
	\end{eqnarray}
	Similarly, we have a stratifiction 
	$${\cal P} := \left\{a\in M(k-1,k-1): a=a^T\geq 0, \tr(a)=1\right\} ~=~ \bigcup_{j=1}^{k-1}{\cal P}^j$$
	of a compact flat convex space 
	${\cal P}$ with non-flat manifolds 
	${\cal P}^j :=\{a\in {\cal P}: \rank(a) = j\}~ (j=1,\ldots,k-1)\,,$ 
	all embedded in $M(k-1,k-1)$. The \emph{Schoenberg map} $\mathfrak{s} : R\Sigma_m^k \to {\cal P}$ is then defined on each stratum by
	$$ \begin{array}{rcl}\mathfrak{s}|_{(R\Sigma_m^k)^j}=:\mathfrak{s}^j: (R\Sigma_m^k)^j &\to& {\cal P}^j \,,~~[x]~\mapsto~x^Tx\end{array}\,.$$
	For $x\in S_j^k$ recall the tangent space decomposition
	$T_xS_j^k = T_{x}[x]\oplus H_xS_j^k$ into the \emph{vertical tangent space} along the orbit $[x]$ and its orthogonal complement the \emph{horizontal tangent space}. For $x\in (S_j^k)^j$ identify canonically (cf. \citep[p. 109]{KBCL99}):
	$$ T_{[x]} (R\Sigma_j^k)^j\cong H_xS_j^k = \{w \in M(j,k-1): \tr(wx^T) = 0, wx^T=xw^T\}\,.$$
	Then the assertion of the following Theorem condenses results of \cite{BandPat05}, cf. also \cite{DKLW08}.
	
	\begin{Th} Each $\mathfrak{s}^j$ is a diffeomorphism with inverse $(\mathfrak{s}^j)^{-1}(a) = [(\sqrt{\lambda} u^T)_1^j]$ where $a = u\lambda u^T$ with $u \in O(k-1)$, $\lambda=\diag(\lambda_1,\ldots,\lambda_m)$, and  $0=\lambda_{j+1} =\ldots=\lambda_{k-1}$ in case of $j<k-1$. Here, $(a)_1^j$ denotes the matrix obtained from taking only the first $j$ rows from $a$. For $x\in S_j^k$ and $w\in H_xS_j^k\cong T_{[x]}(R\Sigma_j^k)^j$ the derivative is given by
	$$d(\mathfrak{s}^j)_{[x]}w = x^Tw+w^Tx\,.$$
	\end{Th}

	\begin{Rm}\label{Schoenberg-emb-non-iso:rm} In contrast to the Veronese-Whitney embedding, the Schoenberg embedding is not isometric as the example of
	$$ x = \left(\begin{array}{cc}\cos \phi&0\\0&\sin\phi\end{array}\right)\,,~~ w_1 = \left(\begin{array}{cc}\sin \phi&0\\0&-\cos\phi\end{array}\right),~~ w_2 = \left(\begin{array}{cc}0&\cos \phi\\\sin\phi&0\end{array}\right),~~$$
	 teaches:
	$\|x^Tw_1+w_1^Tx\| = \sqrt{2}\, 2|\cos\phi\sin\phi|,~~ \|x^Tw_2+w_2^Tx\| =\sqrt{2}\,.$ 
	\end{Rm}

	Since ${\cal P}$ is bounded, convex and Euclidean, the classical expectation $\mathbb E(X^TX)\in {\cal P}^j$ for some $1\leq j\leq k-1$ 
	of the Schoenberg image $X^TX$ of an arbitrary random reflection shape $[X]\in R\Sigma_m^k$ is well defined. Then we have at once the following relation between the rank of the Euclidean mean and increasing sample size. 

	\begin{Th}\label{dim-Schoenberg-mean:thm}
	Suppose that a random reflection shape $[X]\in R\Sigma_m^k$ is distributed absolutely continuous w.r.t. the projection of the spherical volume on $S_m^k$. Then 
	$$\mathbb E(X^TX)\in {\cal P}^{k-1}\mbox{ and~~~} \frac{1}{n}\sum_{i=1}^n X_i^TX_i\in {\cal P}^{\min( nm , k-1)}~ a.s.$$
	for every i.i.d. sample $X_1,\ldots,X_n\sim X$. 
	\end{Th}

	Hence, in  stastical settings involving a higher number of landmarks, a sufficiently well behaved projection of a high rank Euclidean mean onto lower rank ${\cal P}^r\cong (\Sigma_r^k)^r$, usually $m=r$, is to be employed, giving at once a mean shape satisfying strong consistency and a CLT. Here, unlike to intrinsic or Procrustes analysis, the dimension $r$ chosen is crucial for the dimensionality of the mean obtained.

	The orthogonal projection 
	$$\begin{array}{rcl}\phi^r : \bigcup_{i=r}^{k-1} {\cal P}^i &\to&  {\cal P}^r 
	\,,~~ a~\mapsto~\argmin_{b \in {\cal P}^r} \tr\big((a - b)^2\big)\end{array}$$
	giving the set of \emph{extrinsic Schoenberg means} has been computed by \cite{B08}:

	\begin{Th}\label{og-proj-sb:thm} For $1\leq r\leq k-1$, $a = u\lambda u^T \in {\cal P}$ with $u \in O(k-1)$, $\lambda=\diag(\lambda_1,\ldots,\lambda_m)$, $\lambda_{1}\geq \ldots\geq \lambda_{k-1}$ and $\lambda_r>0$ the orthogonal projection onto ${\cal P}^r$ is given by
	$ \phi^r(a) = u \mu u^T$
	with $\mu = \diag(\mu_1,\ldots,\mu_r,0,\ldots,0)$, 
		$$\mu_i = \lambda_i +\frac{1}{r} - \overline{\lambda}_r\,(i=1,\ldots,r)$$
	and $\overline{\lambda}_r = \frac{1}{r}\sum_{i=1}^r\lambda_i \leq \frac{1}{r}$ which is uniquely determined if and only if $\lambda_{r}>\lambda_{r+1}$. 
	\end{Th}
 
	With the notation of Theorem \ref{og-proj-sb:thm}, a non-orthogonal \emph{central projection} $\psi^r(a) =  u \nu u^T$ equally well and uniquely determined has been proposed by \cite{DKLW08} with $$\nu = \diag(\nu_1,\ldots,\nu_r,0,\ldots,0),~~ 
	\nu_i = \frac{\lambda_i}{r\overline{\lambda}_r}~~(i=1,\ldots,r)\,.$$
	Orthogonal and central projection are depicted in Figure \ref{Schoenberg_proj_fig}.

	\begin{figure}
	\begin{minipage}{0.51\textwidth}
	\includegraphics[angle=0,width=1\textwidth]{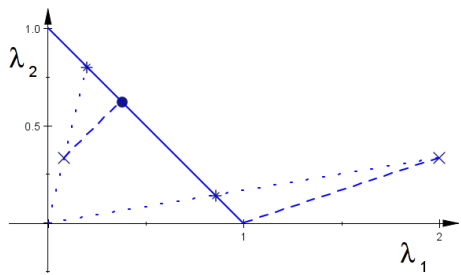}
	\end{minipage}
	\begin{minipage}{0.48\textwidth}
	\caption{\it Projections (if existent) of two points (crosses) in the $\lambda$-plane to the open line segment
	$\Lambda=\{(\lambda_1,\lambda_2): \lambda_1+\lambda_2=1, \lambda_1,\lambda_2>0\}$. The dotted line gives the central projections (denoted by stars) 
	which is well defined for all symmetric, positive 
	  definite matrices (corresponding to the first open quadrant), the dashed line gives the orthogonal projection (circle) which is well defined in the triangle below $\Lambda$ (corresponding to ${\cal P}$) and above $\Lambda$ in an open strip. In particular it exists not for the right point.\label{Schoenberg_proj_fig}}
	\end{minipage}
	\end{figure}

\section{Local Effects of Curvature}\label{local:scn}

	In this section we assume that a manifold stratum $M$ supporting a random element $X$ is isometrically embedded in a Euclidean space $\mathbb R^D$ of dimension $D>0$. With the orthogonal projection $\Phi :\mathbb R^D \to M$ from Section \ref{Frechet-rho-means:scn} and the Riemannian exponential $\exp_p$ of $M$ at $p$ we have the
	$$\begin{array}{lcl}
	 \mbox{\emph{intrinsic tangent space coordinate }}&& \exp_p^{-1}X\mbox{ and the}\\
	 \mbox{\emph{residual tangent space coordinate }}&& d\Phi_p(X-p)\,,
	\end{array}$$
	respectively, of $X$ at $p$, if existent.

\subsection{Finite Power of Tests 
	and Tangent Space Coordinates}\label{tangent:scn} 

	With the above setup, assume that $\mu\in M$ is a unique mean of $X$. Moreover, we assume that $M$ is curved near $\mu$, i.e. that there is $c\in \mathbb R^D$, the center of the osculatory circle touching the geodesic segment in $M$ from $X$ to $\mu$ at $\mu$ with radius $r$. If $X_r$ is the orthogonal projection of $X$ to that circle, then $X=X_r + O(\|X-\mu\|^3)$. Moreover, with
	\begin{eqnarray*} \cos \alpha &=& \left\langle \frac{X-c}{\|X-c\|},  \frac{\mu-c}{r}\right\rangle
	 	~=~ \frac{1}{r^2}\,\langle X_r-c,\mu-c\rangle + O(\|X-\mu\|^3)
	\end{eqnarray*}
	we have the residual tangent space coordinate 
	$$ v = X-c- \frac{\mu-c}{r}\,\|X-c\|\,\cos \alpha =  X_r-c -(\mu-c)\cos \alpha+ O(\|X-\mu\|^3)$$ 
	having squared length $\|v\|^2 =r^2\sin\alpha^2 +   O(\|X-\mu\|^3)$. By isometry of the embedding, the intrinsic tangent space coordinate is given by
	$$\exp^{-1}_{\mu} X = \frac{r\alpha}{\|v\|}\, v + O(\|X-\mu\|^3)\,.$$
	With the component
	$$ \nu = \mu-c - \|X-c\| \frac{\mu-c}{r}\,\cos \alpha = (\mu -c) (1-\cos \alpha)  + O(\|X-\mu\|^3) $$ 
	of $X$ normal to the above mentioned geodesic segment of squared length $\|\nu\|^2 = r^2(1-\cos\alpha)^2+O(\|X-\mu\|^3)$, we obtain $
	 \|\exp^{-1}_{\mu} X\|^2 ~=~ \|v\|^2 + \|\nu\|^2+ O(\|X-\mu\|^3)\,,
	$ since
	\begin{eqnarray*}
	 (1-\cos\alpha)^2 + \sin^2\alpha &=& 2(1-\cos\alpha) ~=~\alpha^2 + 2\frac{\alpha^4}{4!} + \cdots
	\end{eqnarray*}
	and $\alpha =  O(\|X-\mu\|)$.
	In consequence we have

	\begin{Rm}\label{use_is_ext_not_intr_coord_rm}
	In approximation, the variation of intrinsic tangent space coordinates is the sum of the variation $\|v\|^2$ of residual tangent space coordinates and the variation normal to it. In particular, for spheres 
	$$  \|\exp^{-1}_{\mu} X\|^2 ~\geq~ \|v\|^2 + \|\nu\|^2\,.$$ 
	Since the variation in normal space is irrelevant for a two-sample test for equality of means, say, a higher power for tests based on intrinsic means can be expected when solely residual tangent space coordinates obtained from an isometric embedding are used rather than intrinsic tangent space coordinates.

	Note that the natural tangent space coordinates for Ziezold means are residual.
	\end{Rm}

	A simulated classification example in Section \ref{class_data_simulations_scn} illustrates this effect. 

\subsection{The $1:3$ - Property for Spherical and Kendall Shape Means}\label{1:3:scn}

	In this section $M=S^{D-1}\subset \mathbb R^D$ is the $(D-1)$-dimensional unit-hypersphere embedded isometrically in Euclidean $D$-dimensional space. The orthogonal projection $\Phi: \mathbb R^D \to S^{D-1}: p \to \frac{p}{\|p\|}$ is well defined except for the origin $p=0$, and the normal space at $p\in S^{D-1}$ is spanned by $p$ itself. In consequence, a random point $X$ on $S^{D-1}$ has 
	$$d\Phi_p(X-p) = X -p \cos \alpha,~~\exp^{-1}_p(X) = \left\{\begin{array}{ll}\frac{\alpha}{\sin \alpha}~ d\Phi_p(X-p)&\mbox{ for }X\neq p\\0&\mbox{ for }X=p\end{array}\right.$$
	as residual and intrinsic, resp., tangent space coordinate at $-X\neq p\in S^{D-1}$ where $\cos\alpha = \langle X,p\rangle$, $\alpha \in [0,\pi)$. 



	\begin{Th}\label{unique-intr-mean:th} If $X$ a.s. is contained in an open half sphere, 
	it has a unique intrinsic mean which is assumed in the interior of that half sphere.
	\end{Th}
 
	\begin{proof} Below, we show that every intrinsic mean necessarily lies within the interior of the half sphere. Then, \citep[Theorem 7.3]{KWS90}  yields uniqueness. W.l.o.g. let $X = (\sin\phi, x_2,\ldots, x_n)$ such that ${\Prob}\{\sin\phi \leq 0\} = 0 =1- {\Prob}\{\sin\phi >0\}$ and assume  that $p = (\sin \psi, p_2,\ldots, p_n)\in S^{D-1}$ is an intrinsic mean, $-\pi/2\leq \phi,\psi \leq \pi/2$. Moreover let $p' = (\sin (|\psi|), p_2,\ldots, p_n)$. Since
	\begin{eqnarray*}
	\E\left(\|\exp^{-1}_p(X)\|^2\right) &=& \E\left(\arccos^2 \langle p,X\rangle\right)\\
	&=&\E\left(\arccos^2 \left(\sin\psi\sin\phi + \sum_{j=2}^np_jx_j\right)\right)\\
	&\geq& \E\left(\|\exp^{-1}_{p'}(X)\|^2\right)
	\end{eqnarray*}
	with equality if and only if $\sin |\psi| =\sin \psi$, this can only happen for $\sin\psi\geq 0$. Now, suppose that $p=(0,p_2,\ldots,p_n)$ is an intrinsic mean. For small deterministic $\psi \geq 0 $ consider $p(\psi) = (\sin \psi,p_1\cos\psi,\ldots,p_n\cos\psi)$. Then
	\begin{eqnarray*}
	\E\left(\|\exp^{-1}_{p(\psi)}(X)\|^2\right) 
		&=&\E\left(\arccos^2 \left(\sin\psi\sin\phi + \cos \psi \sum_{j=2}^np_jx_j\right)\right)\\
		&=& \E\left(\|\exp^{-1}_{p}(X)\|^2\right) - C_1\psi + O(\psi^2)
	\end{eqnarray*}
	with $C_1 >0$ since ${\Prob}\{\sin\phi >0\}>0$. In consequence, $p$ cannot be an intrinsic mean. Hence, we have shown that every intrinsic mean is contained in the interior of the half sphere.
	\end{proof}

	\begin{Rm}\label{Karcher-means:rm}
	For the special case of spheres, this is a simple proof for the general theorem recently established by \cite{Afsari10} which extends results of  \cite{Ka77,KWS90} and \cite{L01,L04}, stating that the intrinsic mean on a general manifold is unique if among others the support of the distribution is contained in a geodesic half ball.
	\end{Rm}

	The following theorem characterizes the three spherical means.

	\begin{Th}\label{conditions_for_spherical_means_th} Let $X$ be a random point on $S^{D-1}$. Then
	 $x^{(e)} \in S^{D-1}$ is the unique extrinsic mean if and only if the Euclidean mean $\mathbb E(X) =  \int_{S^{D-1}} X\,d\Prob_X$ is non-zero. 
	In that case 
	 $$\lambda^{(e)} x^{(e)} = \mathbb E(X)\,$$
	with $\lambda^{(e)} = \|\mathbb E(X)\|>0$. Moreover, there are suitable $\lambda^{(r)}>0$ and $\lambda^{(i)}>0$ such that every residual mean $x^{(r)} \in S^{D-1}$ satisfies 
	 $$\lambda^{(r)} x^{(r)} = \mathbb E\big(\langle X,x^{(r)}\rangle\, X\big)\,,$$
	 and every intrinsic mean $x^{(i)} \in S^{D-1}$ satisfies 
	 $$\lambda^{(i)} x^{(i)} = \mathbb E\left(\frac{\arccos\langle X,x^{(i)}\rangle}{\sqrt{1-\langle X,x^{(i)}\rangle^2 }}\, X\right)\,.$$
	In the last case we additionally require that $\E\left(\frac{\arccos\langle X,x^{(i)}\rangle}{\sqrt{1-\langle X,x^{(i)}\rangle^2 }}\, \langle X,x^{(i)}\rangle\right)> 0$ which is in particular the case if $X$ is a.s. contained in an open half sphere. 
	\end{Th}

	\begin{proof} The assertions for the extrinsic mean are well known from \cite{HLR96}. The second assertion for residual means follows from minimization of 
	$$ \int_{S^{D-1}} \|p-\langle p,x\rangle x\|^2\,d\Prob_X(p) = 1 - \int_{S^{D-1}} \langle p,x\rangle^2\,d\Prob_X(p)$$
	with respect to $x\in \mathbb R^D$ under the constraining condition $\|x\|=1$. Using a Lagrange ansatz this leads to the necessary condition
	$$ \int_{S^{D-1}} \langle p,x\rangle\, p\,d\Prob_X(p) = \lambda x$$
	with a Lagrange multiplier $\lambda$ of value $\E(\langle X,x\rangle^2)$ which is positive unless $X$ is supported by the hypersphere orthogonal to $x$. In that case, by Proposition \ref{res_means_sphere:prop}, $x$ cannot be a residual mean of $X$, as every residual mean is as well contained in that hypersphere. Hence, we have $\lambda^{(r)}:=\lambda >0$. 

	The Lagrange method applied to
	$$ \int_{S^{D-1}} \|\exp^{-1}_x(p)\|^2 \,d\Prob_X(p) = \int_{S^{D-1}} \arccos^2(\langle p,x\rangle)\,d\Prob_X(p) $$
	taking into account Theorem \ref{unique-intr-mean:th}, insuring that $x^{(i)}$  is in the open half sphere that contatins $X$ a.s., yields the third assertion on the intrinsic mean.
	\end{proof}

	Recall that residual means are eigenvectors to the largest eigenvalue of the matrix $\mathbb E(XX^T)$. As such, they rather reflect the mode than the classical mean of a distribution:

	\begin{Ex}\label{different_means_ex}
	 Consider $\gamma \in (0,\pi)$ and a random variable $X$ on the unit circle $\{e^{i\theta}: \theta \in [0,2\pi)\}$ which takes the value $1$ with probability $2/3$ and $e^{i\gamma}$ with probability $1/3$. Then, explicit computation gives the unique intrinsic and extrinsic mean as well as the two residual means
	$$ x^{(i)} = e^{i\frac{\gamma}{3}},\quad x^{(e)} = e^{i\arctan\frac{\sin \gamma}{2+\cos\gamma}},\quad x^{(r)} = \pm\, e^{i\frac{1}{2}\,\arctan\frac{\sin (2\gamma)}{2+\cos(2\gamma)}}\,.$$
	Figure \ref{different_means_fig} shows the case $\gamma = \frac{\pi}{2}$.
	\end{Ex}
 
	\begin{figure}
	\begin{minipage}{0.5\textwidth}
	\includegraphics[width=0.95\textwidth]{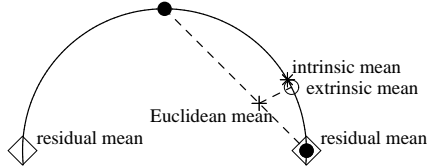}
	\end{minipage}
	\begin{minipage}{0.49\textwidth}
 	\caption{\it Means on a circle of a distribution taking the upper dotted value with probability $1/3$ and the lower right dotted value with probability $2/3$. The latter happens to be one of the two residual means.\label{different_means_fig}}
	\end{minipage}
	\end{figure}

	In contrast to Figure \ref{different_means_fig}, one may assume in many practical applications that the mutual distances of the unique intrinsic mean $x^{(i)}$, the unique extrinsic mean $x^{(e)}$ and the unique residual mean $x^{(r_0)}$ closer to $x^{(e)}$ are rather small, namely of the same order as the squared proximity of the modulus  
 	$\|\mathbb E(X)\|$  of the Euclidean mean to $1$, cf. Table \ref{1:3:tab}. We will use the following condition
	\begin{eqnarray}\label{reg_extr_mean_cond}\left.\begin{array}{rcl}
	 \|x^{(e)} -x^{(r_0)}\|\,,~~  \|x^{(e)} -x^{(i)}\|&=& 
	O\big((1-\|\E(X)\|)^2\big) 
	\end{array}\right. 
	\end{eqnarray}
	with the \emph{concentration parameter} $1-\|\E(X)\|$.

	\begin{Cor}\label{3_means_on_circle_cor}
	 	Under condition (\ref{reg_extr_mean_cond}), if all three means are unique, then the great circular segment between the residual  mean $x^{(r_0)}$ closer to the extrinsic mean $x^{(e)}$ and the intrinsic mean $x^{(i)}$ is divided by the extrinsic mean in approximation by the ratio $1:3$:
		\begin{eqnarray*}
	x^{(r_0)} &=& \frac{\|\E(X)\|}{\lambda^{(r)}}\left(x^{(e)} - \frac{\mathbb E\big( \langle X-x^{(e)} ,X\rangle\,X\big)}{\|\E(X)\|} + O\big((1-\|\E(X)\|)^2\big)\right)\\
	x^{(i)} &=& \frac{\|\E(X)\|}{\lambda^{(i)}}\left(x^{(e)} + \frac{1}{3}\, \frac{\mathbb E\big( \langle X-x^{(e)} ,X\rangle\,X\big)}{\|\E(X)\|} + O\big((1-\|\E(X)\|)^2\big)\right)\,
		\end{eqnarray*}
	with $\lambda^{(i)}$ and $\lambda^{(r)}$ from Theorem \ref{conditions_for_spherical_means_th}.
	\end{Cor}

	\begin{proof}	
	 For any $x,p \in S^{D-1}$ decompose
	$ p - x = p - \langle x,p\rangle\, x - z(x,p) x $ with   $z(x,p) = 1 - \langle x,p\rangle$, the length of the part of $p$ normal to the tangent space at $x$. Note that $\E\big(z(x^{(e)},X\big) = 1-\|\E(X)\|$.
	Now, under condition (\ref{reg_extr_mean_cond}), verify the first assertion using Theorem \ref{conditions_for_spherical_means_th}:
	\begin{eqnarray*}\label{approx_extr_res_mean}\nonumber
	 x^{(r0)} &=&\frac{1}{\lambda^{(r)}}\left(\E(X) - \E\big(z(x^{(r_0)},X)X\big)\right)\,.
	\end{eqnarray*}
	On the other hand since
	\begin{eqnarray*}
	 \frac{\arccos (1-z)}{\sqrt{1-(1-z)^2 }} &=& 1 + \frac{1}{3}\,z +\frac{2}{15}\,z^2 +\ldots
	\end{eqnarray*}
	we obtain with the same argument the second assertion
	\begin{eqnarray*} 
	 x^{(i)} &=& \frac{1}{\lambda^{(i)}} \E \left(\frac{\arccos\langle X,x^{(i)}\rangle}{\sqrt{1-\langle X,x^{(i)}\rangle^2 }}\, X\right)\\ 
	&=& \frac{1}{\lambda^{(i)}}\left(\E(X)+ \frac{1}{3}\, \E\left( z(x^{(i)},X)\,X\right) + \frac{2}{15}\, \E\left( z(x^{(i)},X)^2\,X\right) +\ldots\right)\\ 
	&=& \frac{\|\E(X)\|}{\lambda^{(i)}}\left(x^{(e)} +\frac{1}{3\|\E(X)\|}\, \E \left(z(x^{(e)},X)\,X\right) + O\big((1-\|\E(X)\|)^2\big)\right)\,. 
	\end{eqnarray*}
	\end{proof}

	\begin{Rm} The tangent vector defining the great circle approximately connecting the three means is obtained from correcting with the expected normal component of any of the means. As numerical experiments show, this great circle is different from the first \emph{principal component geodesic} as defined in \cite{HZ06}. 
	\end{Rm}

	Recall the following connection between top and quotient space means, 
	cf. Theorem 
	\ref{pop_int_mean_bottom_int_mean_top:Th}. 

	\begin{Rm}\label{kend-mf-quot-mean:rm} Let $p \in S_m^k$ such that a random shape $X$ on $\Sigma_m^k$ is supported by $(\Sigma_m^k)^{([p])}\cup A$ with $A\subset \Sigma_m^k$ at most countable. By Lemma \ref{cut_locus:lem} and Theorem \ref{pop_int_mean_bottom_int_mean_top:Th}, $(\Sigma_m^k)^{([p])}\cup A$ admits a horizontal measurable lift $L\subset S_m^k$ in optimal position to $p\in S_m^k$. Define the random variable $Y$ on $L\subset S_m^k$ by $\pi \circ Y = X$. Then we have that 
	\begin{enumerate}
	 \item[] if   $[p]$ is an intrinsic mean of $X$ then $p$ is an intrinsic mean of $Y$,
	 \item[] if   $[p]$ is a full Procrustean mean of $X$ then $p$ is a residual mean of $Y$,
	 \item[] if   $[p]$ is a Ziezold mean of $X$ then $p$ is an extrinsic mean of $Y$.
	\end{enumerate}
	\end{Rm}

	In consequence, Corollary \ref{3_means_on_circle_cor} extends at once to Kendall's shape spaces. \emph{Generalized geodesics} referred to below are an extension of the concept of geodesics to non-manifold shape spaces, cf. \cite{HHM07}. 

	\begin{Cor}\label{3_means_on_kendall_cor}
	 Suppose that a random shape $X$ on $\Sigma_m^k$ with unique intrinsic mean $\mu^{(i)}$, unique Ziezold mean $\mu^{(z)}$ and unique Procrustean mean $\mu^{(p)}$
	is supported by $R=(\Sigma_m^k)^{(\mu^{(i)})}\cap(\Sigma_m^k)^{(\mu^{(z)})}\cap (\Sigma_m^k)^{(\mu^{(p)})}$. 
	If the 	means are sufficiently close to each other in the sense of
	 $$d_{\Sigma_m^k}(\mu^{(z)} -\mu^{(p)})\,,~~  d_{\Sigma_m^k}(\mu^{(z)} -\mu^{(i)})~=~ 
	O\big((1-\|\E(Y)\|)^2\big) $$ 
	with the random pre-shape $Y$ on a horizontal lift $L$ of $R$ defined by $X = \pi \circ Y$,
	then the generalized geodesic segment between $\mu^{(i)}$ and $\mu^{(p)}$ is approximately divided by $\mu^{(z)}$ by the ratio $1:3$ with an error of order  $O\big((1-\|\E(Y)\|)^2\big)$.
	\end{Cor}

\section{Examples: Exemplary Datasets and Simulations}\label{class_data_simulations_scn}
	All of the results of this section are based on datasets and simulations, i.e., all means considered are sample means. An R-package for the computation of all means including the poplar leaves data can be found under \cite{Hshapes}.

	\subsection{The $1:3$ property}
	In the first example we illustrate Corollary \ref{3_means_on_kendall_cor} on the basis of four classical data sets:
	\begin{description}
	\item[poplar leaves:] contains 104 quadrangular planar shapes extracted from poplar leaves in a joint collaboration with Institute for Forest Biometry and Informatics at the University of G\"ottingen, cf. \cite{Hshapes, HHM09}.
	\item[digits '3':] contains 30 planar shapes with 13 landmarks each, extracted from handwritten digits '3', cf. \citep[p. 318]{DM98}.
	\item[macaque skulls:] contains three-dimensional shapes with 7 landmarks each, of 18 macaque skulls, cf. \citep[p. 16]{DM98}. 
	\item[iron age brooches:] contains 28 three-dimensional tetrahedral shapes of iron age brooches, cf. \citep[Section 3.5]{S96}. 
	\end{description}

	\begin{figure}[h!]\centering
	 \includegraphics[angle=-90,width=0.45\textwidth]{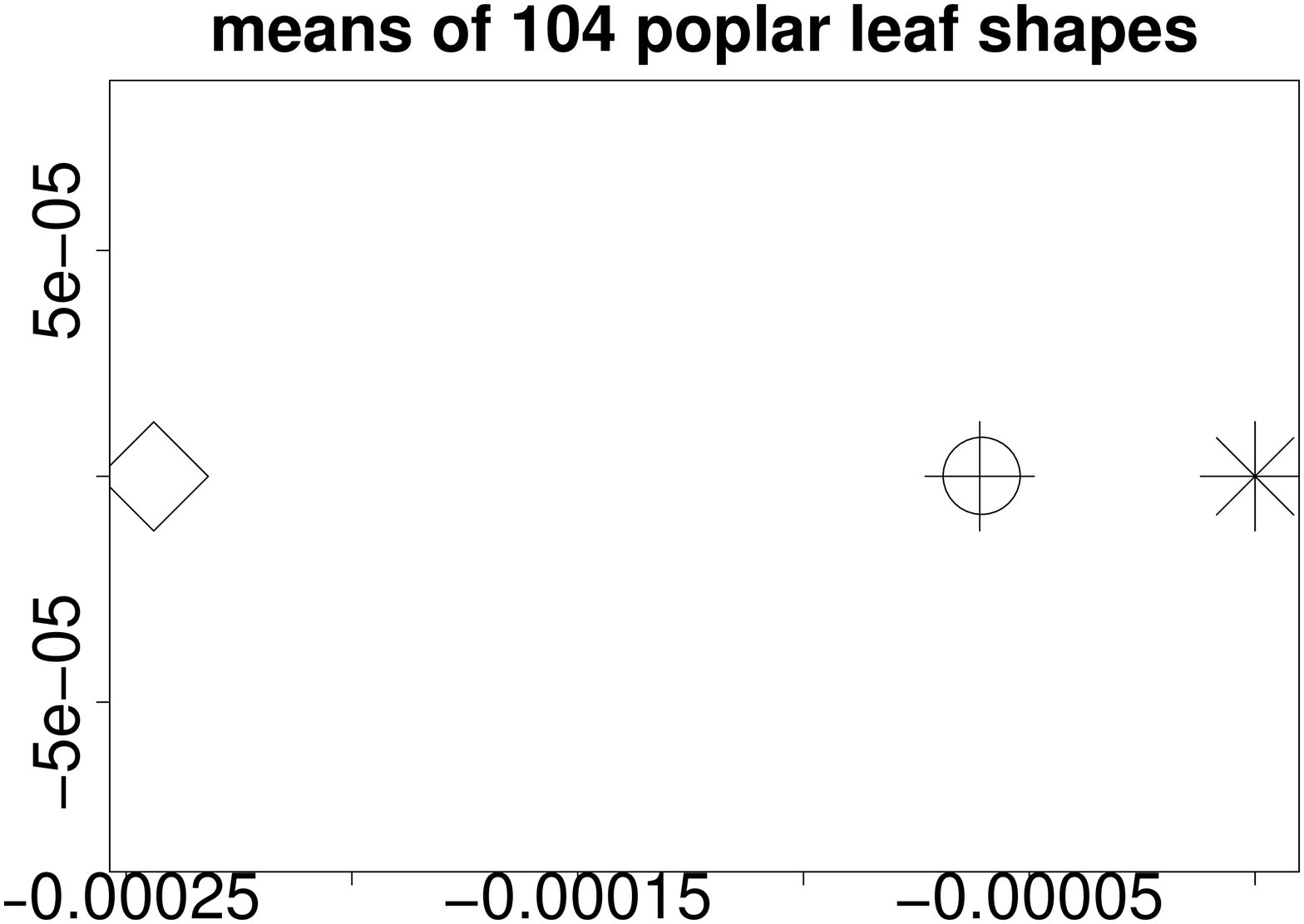}
	 \includegraphics[angle=-90,width=0.45\textwidth]{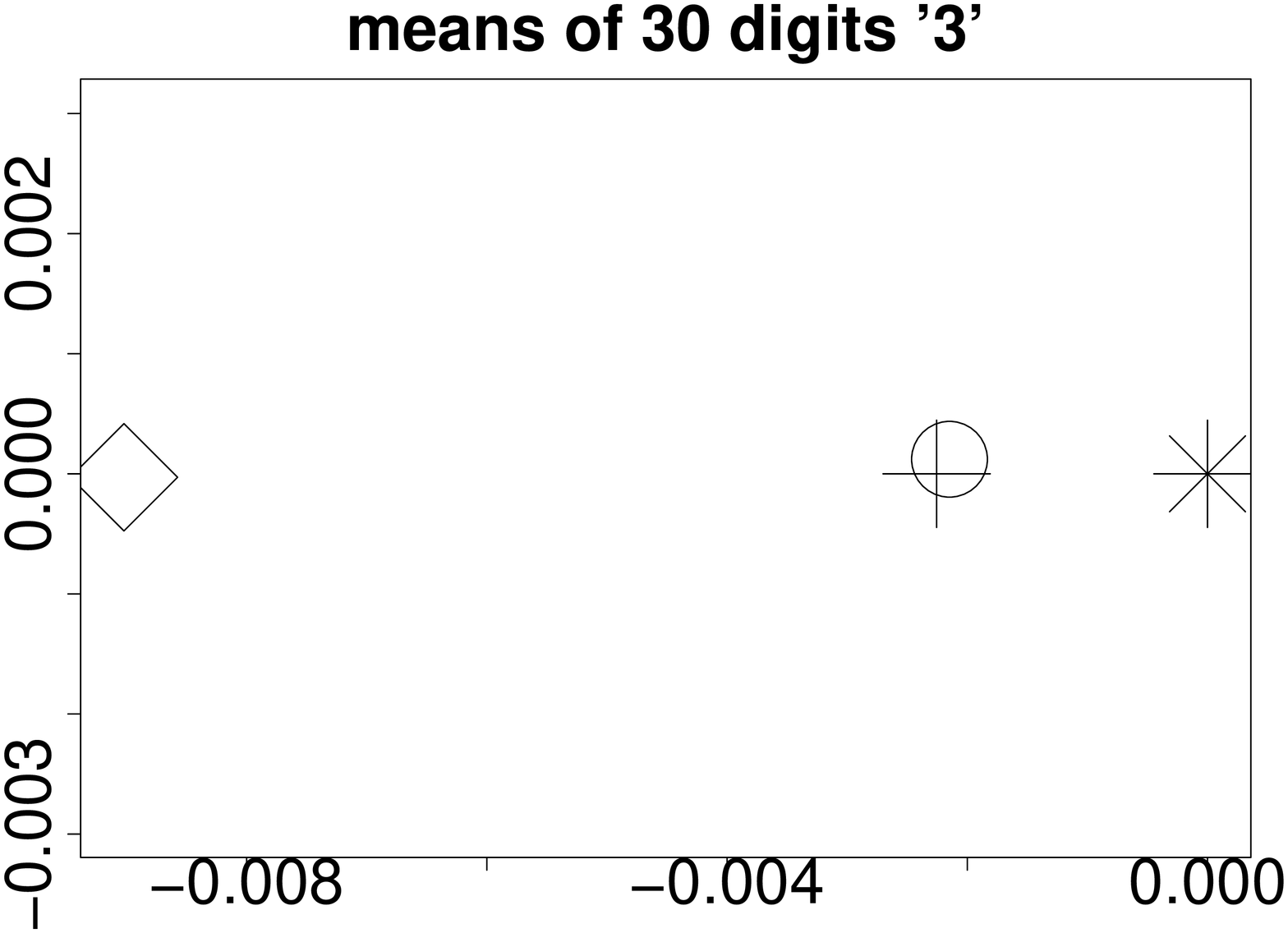}
	\\
	 \includegraphics[angle=-90,width=0.45\textwidth]{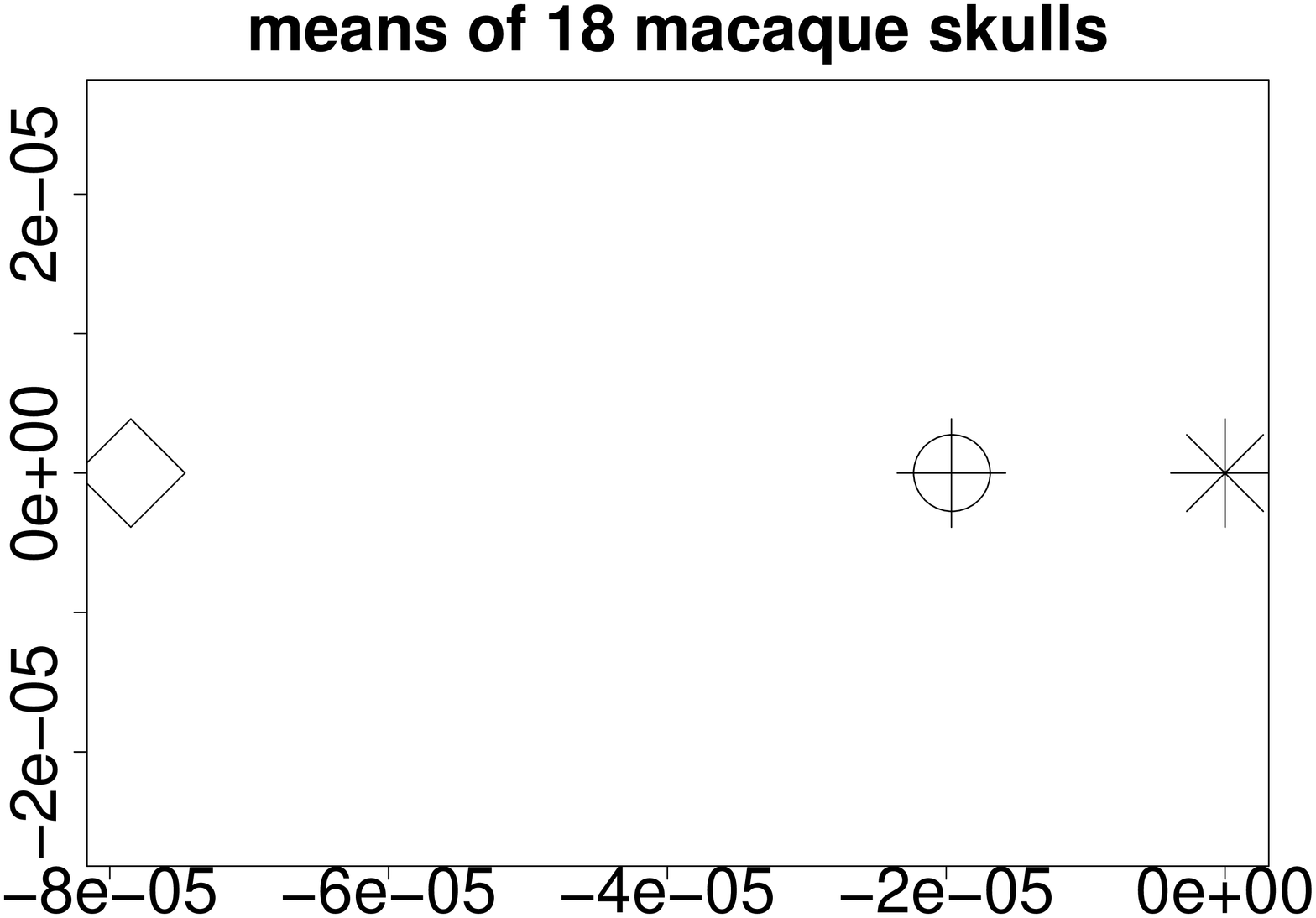}
	 \includegraphics[angle=-90,width=0.45\textwidth]{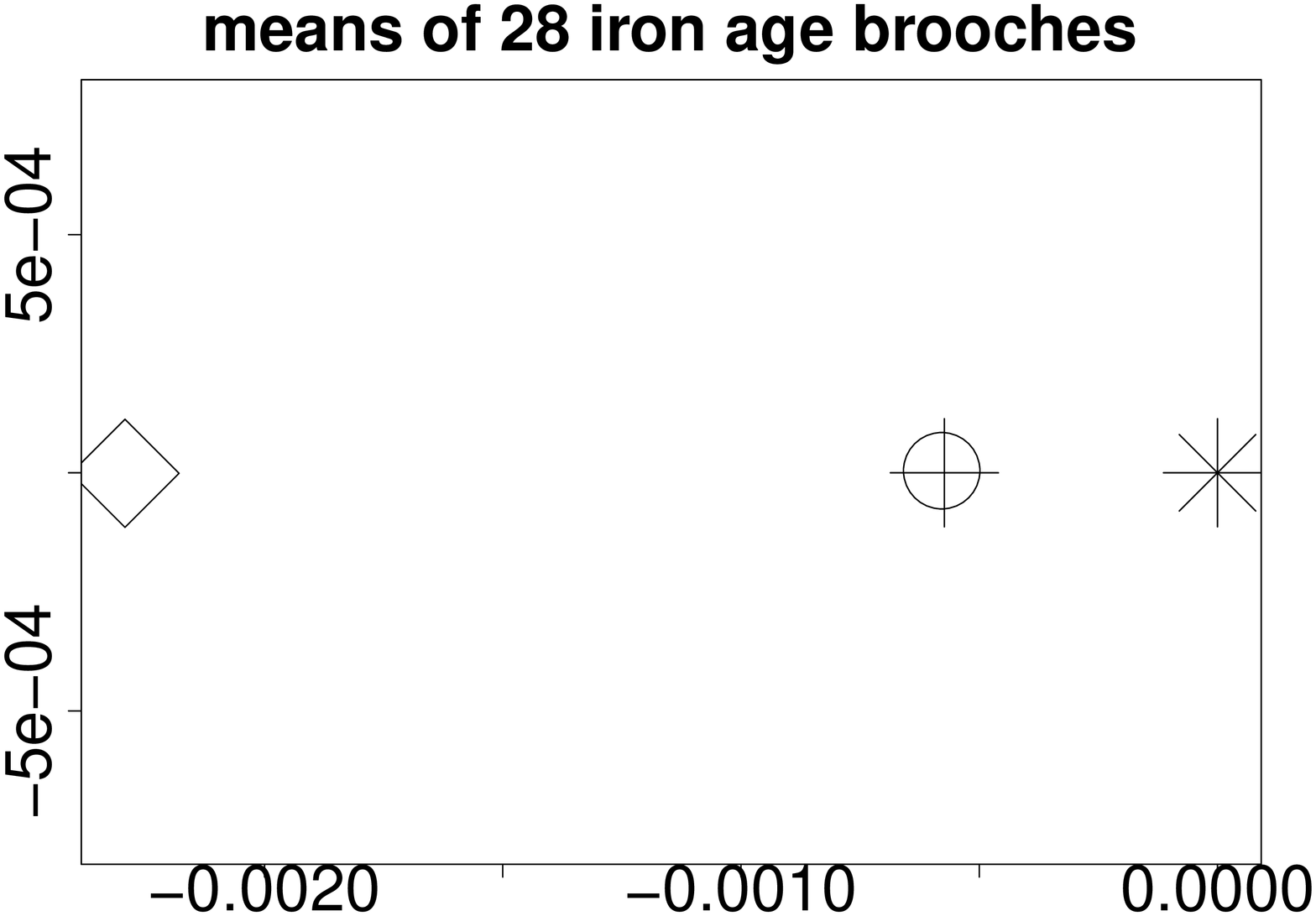}
	\caption{\it Depicting shape means for four typical data sets: intrinsic (star), Ziezold (circle) and full Procrustes (diamond) projected to the tangent space at the intrinsic mean. The cross divides the generalized geodesic segment joining the intrinsic with the full Procrustes mean by the ratio $1:3$.\label{three_means:fig}}
	\end{figure}
	
	 \begin{table}[h!]\centering\fbox
	{
	\footnotesize
	 $\begin{array}{r|ccc|l}
	  \mbox{data set}
	&d_{\Sigma_m^k}(\mu^{(i)},\mu^{(z)})&d_{\Sigma_m^k}(\mu^{(p)},\mu^{(z)})&d_{\Sigma_m^k}(\mu^{(p)},\mu^{(i)})&(1-\|\E(Y)\|)^2\\
	\hline
	\mbox{poplar leaves}&6.05e-05&1.83e-04&2.44e-04&5.24e-05\\
	\mbox{digits '3'}&0.00154&0.00452&0.00605& 0.00155\\
	\mbox{macaque skulls}&1.96e-05&5.89e-05&7.85e-05&7.59e-06\\
	\mbox{iron age brooches}&0.000578&0.001713&0.002291&0.000217
	 \end{array}$}
	\caption{\it Mutual shape distances between intrinsic mean $\mu^{(i)}$, Ziezold mean $\mu^{(z)}$ and full Procrustes mean $\mu^{(p)}$ for various data sets. Last column: the concentration parameter from (\ref{reg_extr_mean_cond}), cf. also Corollary \ref{3_means_on_kendall_cor}.\label{1:3:tab}}
	\end{table}
	
	As clearly visible from Figure \ref{three_means:fig} and Table \ref{1:3:tab}, the approximation of Corollary \ref{3_means_on_kendall_cor} for two- and three-dimensional shapes is highly accurate for data of little dispersion (the macaque skull data) and still fairly accurate for highly dispersed data (the digits '3' data).

	\subsection{``Partial Blindness'' of full Procrustes and Schoenberg Means}\label{blindness:scn}
	In the second example we illustrate an effect of ``blindness to data'' of full Procrustes means and Schoenberg means. The former blindness is due to the affinity of the Procrustes mean to the mode in conjunction with curvature, the latter is due to non-isometry of the Schoenberg embedding. While the former effect occurs only for some highly dispersed data when the analog of condition (\ref{reg_extr_mean_cond}) is violated, the latter effect is local in nature and may occur for concentrated data as well.

	\begin{figure}[h!]
	\begin{minipage}{0.6\textwidth}\centering
	 \includegraphics[angle=-90,width=1\textwidth]{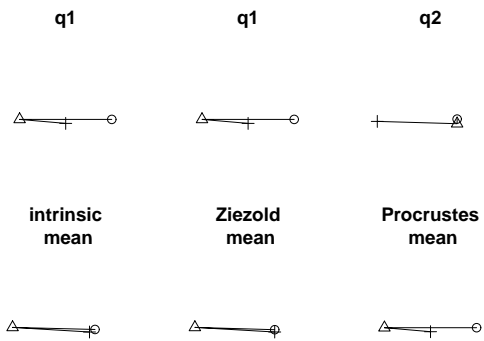}
	\end{minipage}
	\begin{minipage}{0.05\textwidth}\hfill \end{minipage}
	\begin{minipage}{0.33\textwidth}
	\caption{\it A data set of three planar triangles (top row) with its corresponding intrinsic mean (bottom left), Ziezold mean (bottom center) and full Procrustes mean (bottom right).\label{procrustes_blind:fig}}
	\end{minipage}
	\end{figure}

	Reenacting the situation of Section \ref{Procrustes-convex:scn}, cf. also Example \ref{different_means_ex} and Figure \ref{different_means_fig}, the shapes of the triangles $q_1$ and $q_2$ in Figure \ref{procrustes_blind:fig} are almost maximally remote. Since the mode $q_1$ is assumed twice and $q_2$ only once, the full Procrustes mean is nearly blind to $q_2$. 

	\begin{figure}[h!]
	\begin{minipage}{0.6\textwidth}
	\centering
	 \includegraphics[angle=-90,width=0.3\textwidth]{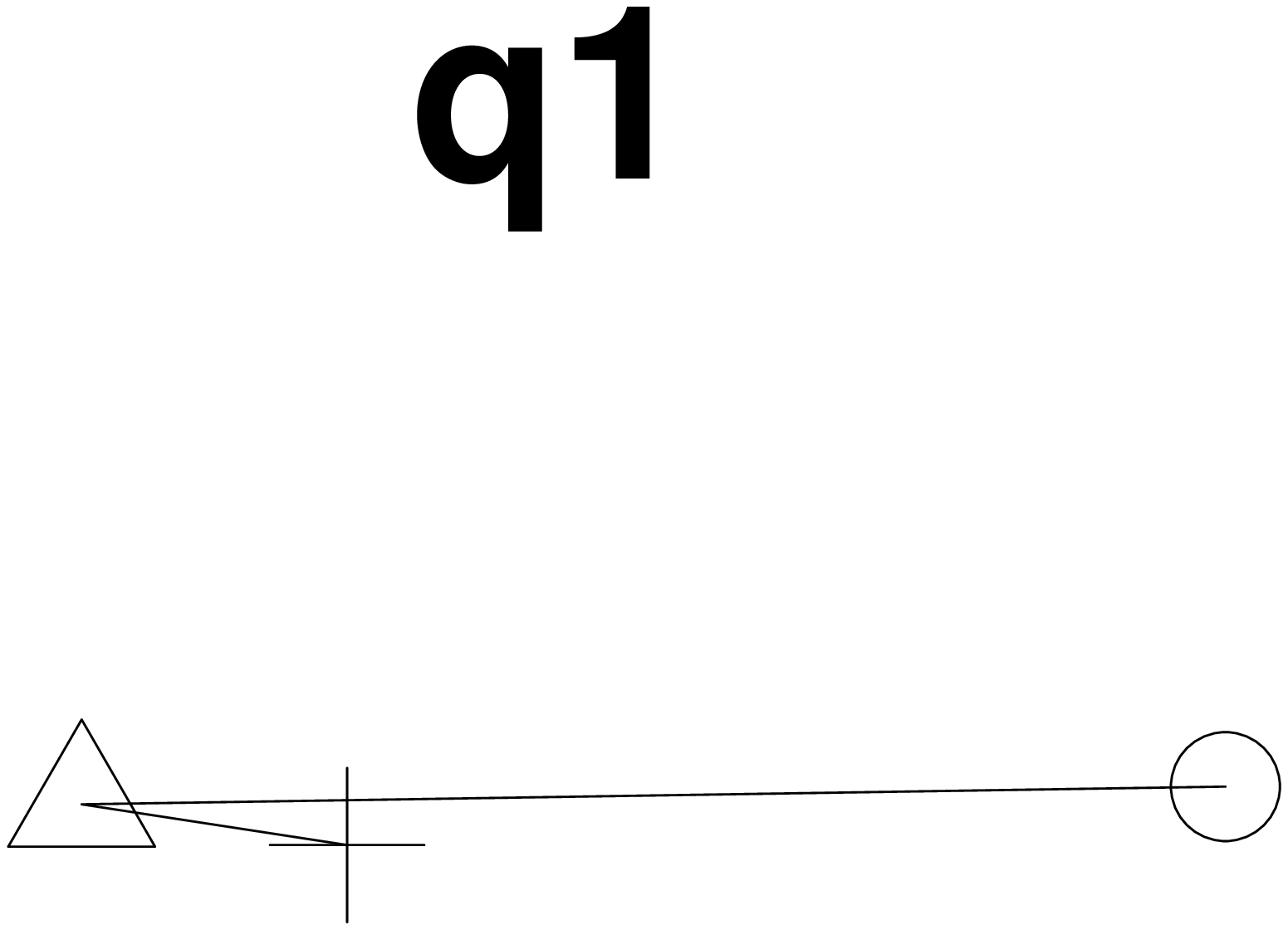}
	 \includegraphics[angle=-90,width=0.3\textwidth]{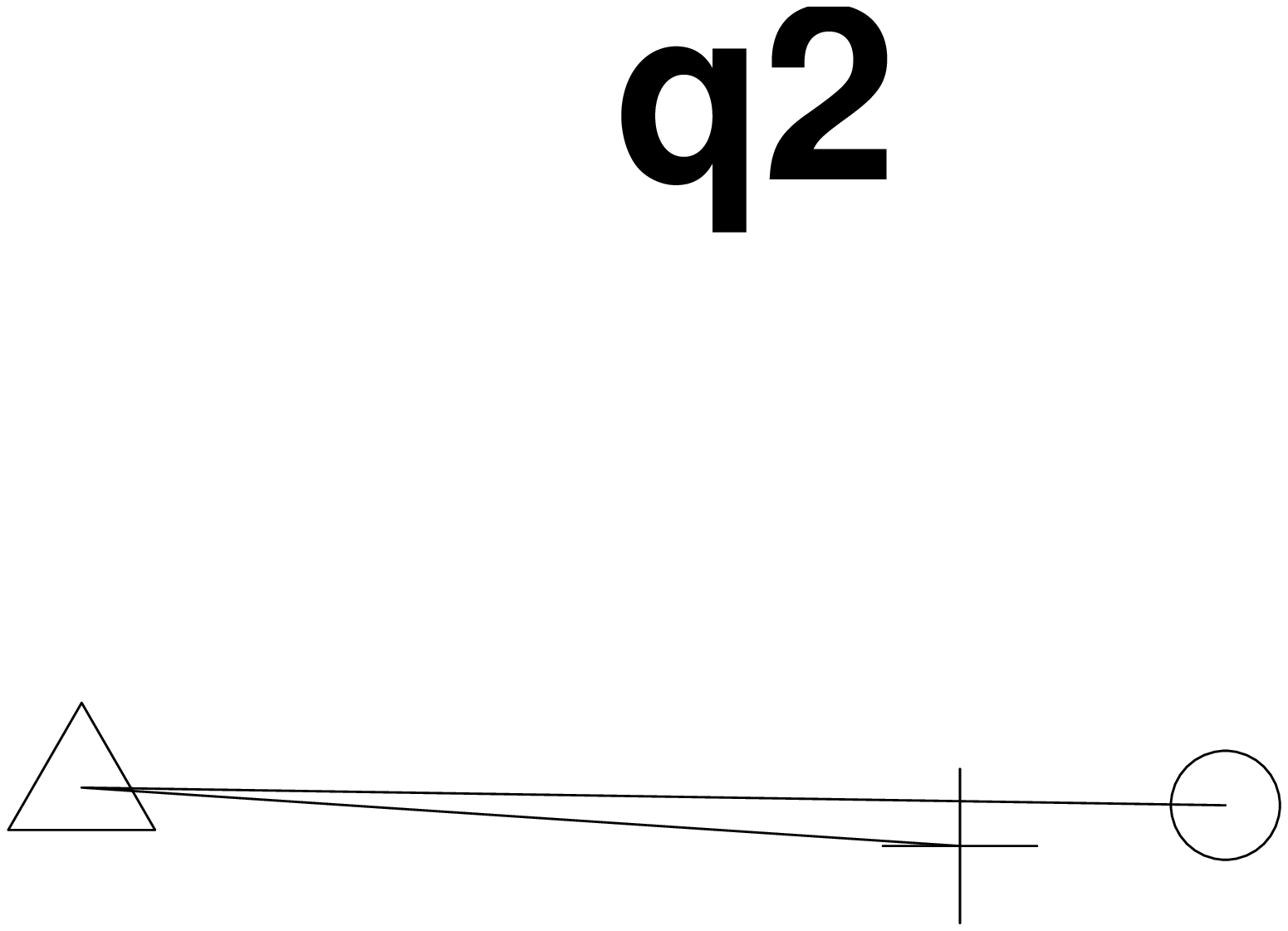}
	 \includegraphics[angle=-90,width=0.3\textwidth]{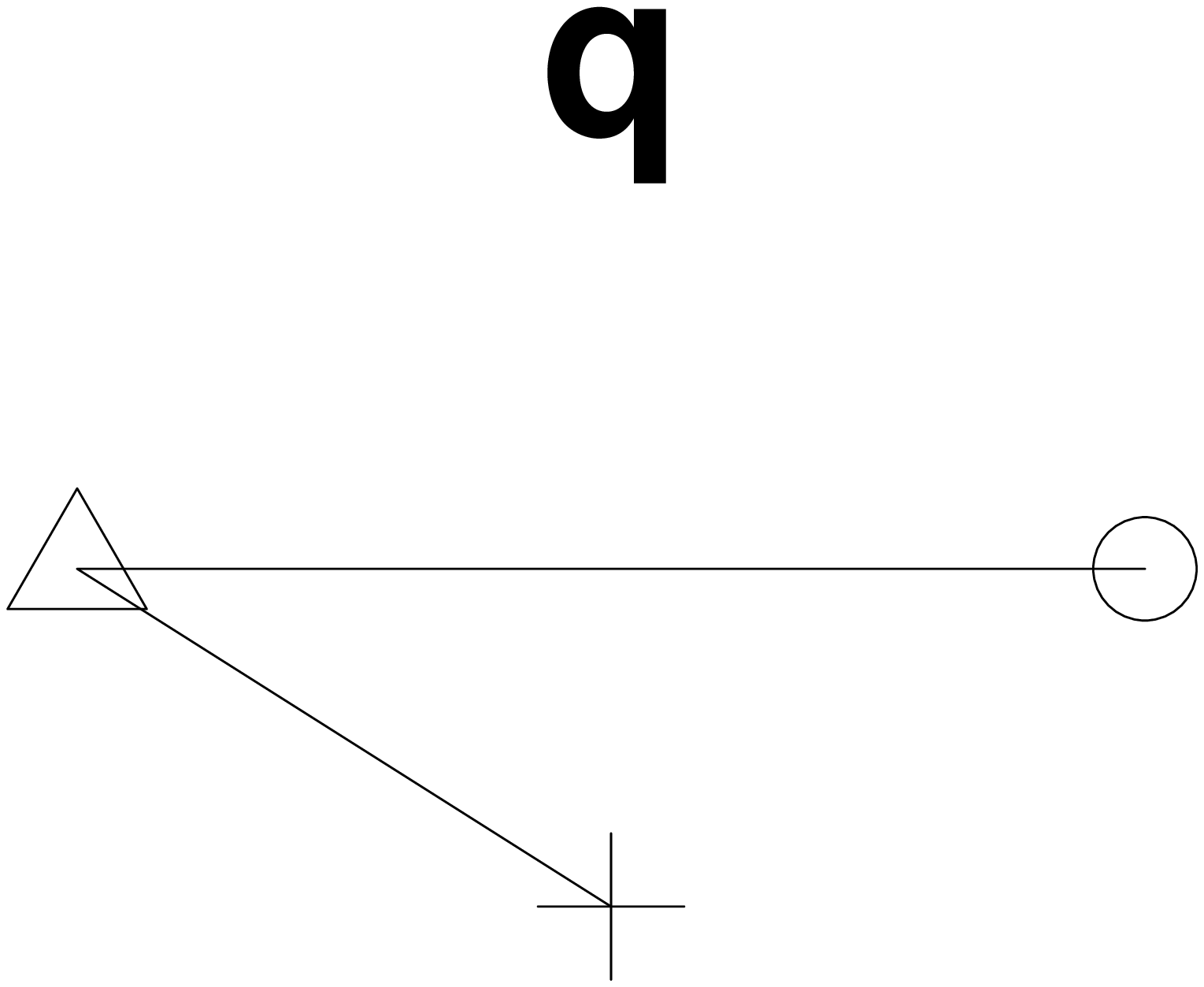}\\
	 \includegraphics[angle=-90,width=1\textwidth]{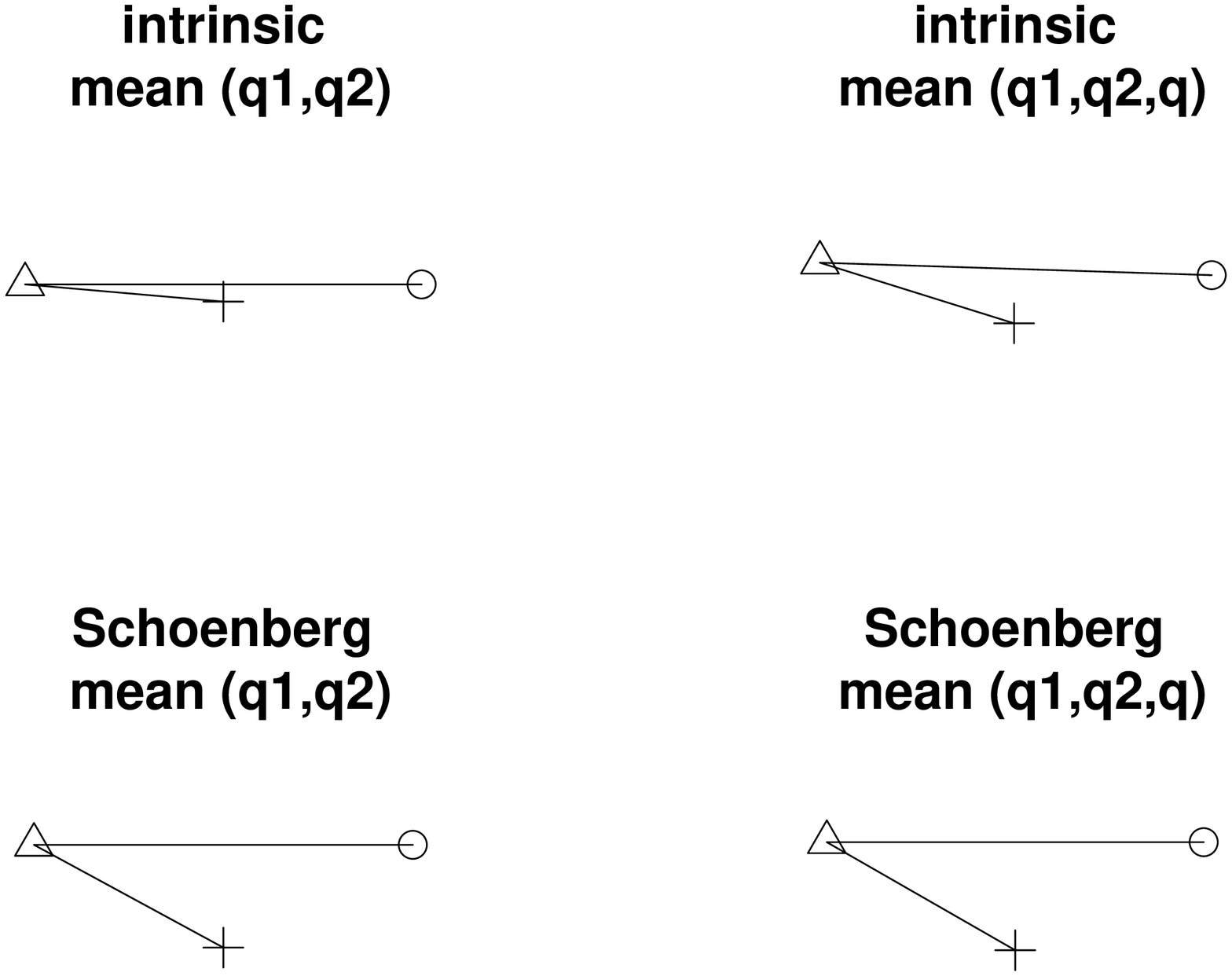}\end{minipage}
	\begin{minipage}{0.05\textwidth}\hfill \end{minipage}
	\begin{minipage}{0.33\textwidth}
	\caption{\it Planar triangles $q_1= x\cos\beta   - w_2\sin\beta$,  $q_2= x\cos\beta   + w_2\sin\beta$ and $q = x\cos\beta   + w_1\sin\beta$ with $x,w_1,w_2$ from Remark \ref{Schoenberg-emb-non-iso:rm}, $\phi =0.05$, $\beta=0.3$ (top row). Intrinsic means (middle row) of sample $(q_1,q_2)$ (left) and  $(q_1,q_2,q)$ (right). Schoenberg means (bottom row) of sample $(q_1,q_2)$ (left) and  $(q_1,q_2,q)$ (right).\label{Schoenberg_blind:fig} }
	\end{minipage}
	\end{figure}

	Even though Schoenberg means have been introduced to tackle 3D shapes, the effect of ``blindness'' can be well illustrated already for 2D. To this end consider $x=x(\phi)$, $w_1=w_1(\phi)$ and $w_2=w_2(\phi)$ as introduced in Remark \ref{Schoenberg-emb-non-iso:rm}. Along the horizontal geodesic through $x$ with initial velocity $w_2$ we pick two points $q_1= x\cos\beta  +w_2\sin\beta$ and $q_2= x\cos\beta   - w_2\sin\beta$. On the orthogonal horizontal geodesic through $x$ with initial velocity $w_1$ pick $q = x\cos\beta'   + w_1\sin\beta'$. Recall from Remark \ref{Schoenberg-emb-non-iso:rm}, that along that geodesic the derivative of the Schoenberg embedding can be made arbitrarily small for $\phi$ near $0$. Indeed, Figure \ref{Schoenberg_blind:fig} illustrates that 
	in contrast to the intrinsic mean, the Schoenberg mean is ``blind'' to the strong collinearity of $q_1$ and $q_2$. 

	\subsection{Discrimination Power}\label{discrimination:scn}
	In the ultimate example we illustrate the consequences of the choice of tangent space coordinates and the effect of the tendency of the Schoenberg mean to increase dimension  
	by a classification simulation. To this end we apply a Hotelling $T^2$-test to discriminate the shapes of 10 noisy samples of regular unit cubes from the shapes of 10 noisy samples of pyramids with top section chopped off, each with 8 landmarks, given by the following configuration matrix
	$$\left(\begin{array}{cccccccc} 0& 1 &\frac{1+\epsilon}{2}& \frac{1-\epsilon}{2}& 0& 1&\frac{1+\epsilon}{2}& \frac{1-\epsilon}{2}\\ 0& 0& \frac{1-\epsilon}{2}&\frac{1-\epsilon}{2} & 1& 1& \frac{1+\epsilon}{2}&\frac{1+\epsilon}{2}\\ 0& 0& \epsilon& \epsilon& 0& 0&\epsilon&\epsilon\end{array} \right)\,$$
	(cf. Figure \ref{cube_pyramid:fig}) determined by $\epsilon >0$. In the simulation, independent Gaussian noise 
	is added to each landmark measurement. Table \ref{sim-tab} gives the percentages of correct classifications. 

	\begin{figure}[h!]
	\centering
	 \includegraphics[width=0.45\textwidth]{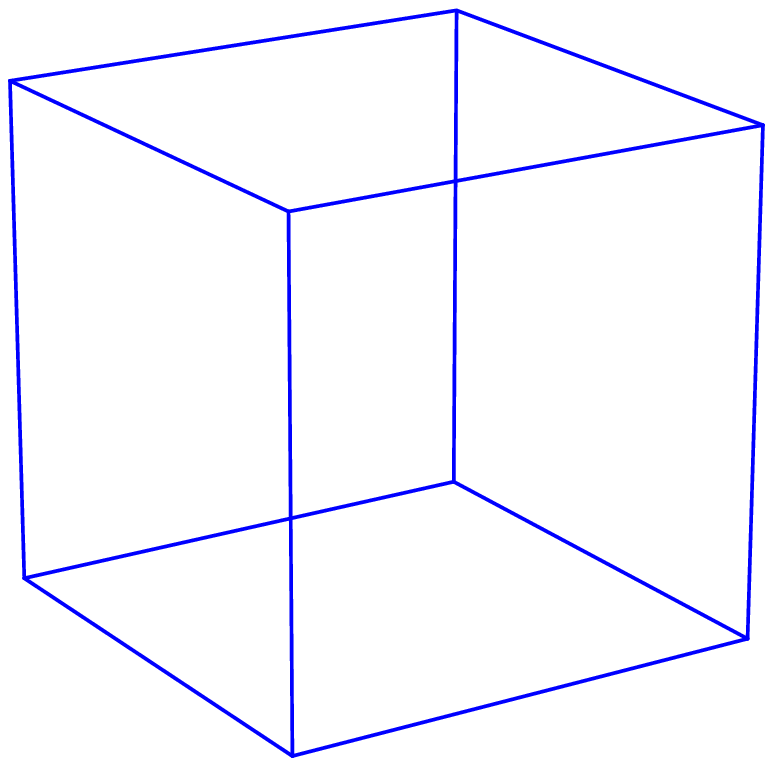}
	 \includegraphics[width=0.45\textwidth]{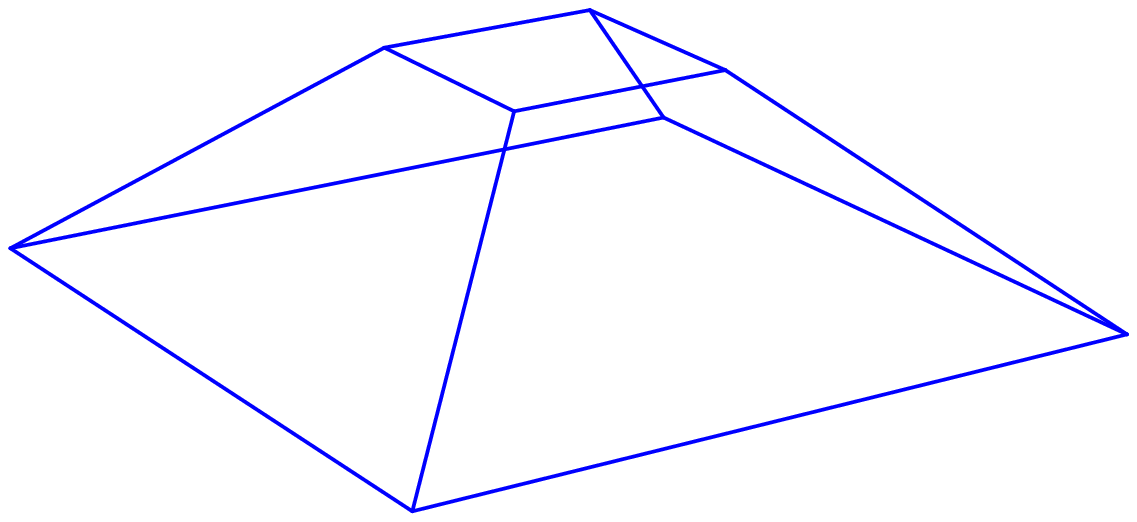}
	\caption{\it cube (left) and pyramid of varying height $\epsilon$ (right) for classification.\label{cube_pyramid:fig}}
	\end{figure}

	\begin{table}[h!]\centering\fbox{
	$\begin{array}{c|cccc}
         &\multicolumn{2}{c}{\mbox{intrinsic mean with}}&\mbox{Ziezold}&\mbox{Schoenberg}\\
         \epsilon&\mbox{intrinsic}&\mbox{residual}&\mbox{mean}&\mbox{mean}\\
	&\multicolumn{2}{c}{\mbox{tangent space coordinates}}\\
	\hline
        0.0 &70\,\% & 74\,\% & 74\,\% & 64\,\%\\	
	 0.2  &56\,\% & 58\,\% & 57\,\% & 51\,\%\\
	0.3  &41\,\% & 42\,\% & 42\,\% & 42\,\% 
	\end{array}$}
 	\caption{\it Percentage of correct classifications within 1,000 simulations each of $10$ unit-cubes and $10$ pyramids determined by $\epsilon$ (which gives the height), where each landmark is independently corrupted by Gaussian noise of variance $\sigma^2=0.2$ via a Hotelling T$^2$-test for equality of means to the significance level $0.05$.\label{sim-tab}}
	\end{table}

	As visible from Table \ref{sim-tab}, discriminating flattened pyramids ($\epsilon = 0$) from cubes ($\epsilon=1$) is achieved much better by employing intrinsic or Ziezold means rather than 
	Schoenberg means. This finding is in concord with Theorem \ref{dim-Schoenberg-mean:thm}: samples of size $10$ of two-dimensional configurations yield Euclidean means a.s. in ${\cal P}^{7}$ which are projected to ${\cal P}^{3}$ to obtain Schoenberg means in $\Sigma_3^8$. In consequence, Schoenberg means of noisy nearly two-dimensional pyramids are essentially three dimensional. With increased height of the pyramid, ($\epsilon>0$, 
	i.e. for more pronounced third dimension and increased proximity to the unit cube) this effect waynes and all means perform equally well (or bad). 
	Moreover in any case, intrinsic means with intrinsic tangent space coordinates qualify less for shape discrimination than intrinsic means with residual tangent space coordinates, cf. Remark \ref{use_is_ext_not_intr_coord_rm}. The latter (intrinsic means with residual tangent space coordinates) are better or equally well behaved as Ziezold means (which naturally use residual tangent space coordinates).

	\begin{table}[h!]\centering \fbox{
	 $\begin{array}{ccc}
	   \mbox{intrinsic mean}&\mbox{Ziezold mean} &\mbox{Schoenberg mean}\\\hline 
	0.24& 0.18& 0.04\end{array}$}
	\caption{\it Average time  in seconds for the computation of means in $\Sigma_3^8$ of sample size $20$ on a PC with a $800$ MHZ CPU based on $1,000$ repetitions.\label{comp_time:tab}}
	\end{table}
	In conclusion we record the time for the computations of means in Table \ref{comp_time:tab}. While Ziezold means compute in approximately $3/4$ of the computational time for intrinsic means, Schoenberg means are obtained approximately $6$ times faster.


\section{Discussion}\label{discussion-sec} 

	By establishing stability results for intrinsic and Ziezold means on the manifold part of a shape space, a gap in asymptotic theory for general non-manifold shape spaces could be closed, now allowing for multi-sample tests of equality of intrinsic means and Ziezold means. A similar stability assertion in general is false for Procrustean means for low concentration. There is reason to believe, however, that it would be true for higher concentration. 
	Note that the argument applied to intrinsic and Ziezold means fails for Procrustean means, since in contrast to the equations in Theorem \ref{pop_int_mean_bottom_int_mean_top:Th} 
	the sum of Procrustes residuals is in general non-zero. 
	Loosely speaking, the findings on dimensionality condense to
	{\it \begin{itemize}
	\item[--] Procrustean means may decrease dimension by 2 or more,
	\item[--] intrinsic and Ziezold means decrease dimension at most by 1, in particular, they preserve regularity,
	\item[--] Schoenberg means tend to increase up to the maximal dimension possible.
	\end{itemize}}

	Due to the proximity of Ziezold and intrinsic means on Kendall's shape spaces in most practical applications, taking into account that the former are computationally easier accessable (optimally positioning and Euclidean averaging in every iteration step) than intrinsic means (optimally positioning and weighted averaging in every iteration step), 
	Ziezold means can be preferred over intrinsic means. They may be even more preferred over intrinsic means, since  Ziezold means naturally come with residual tangent space coordinates which may allow in case of intrinsic means for a higher finite power of tests than intrinsic tangent space coordinates.

	Computationally much faster (not relying on iteration at all) are Schoenberg means which are available for Kendall's reflection shape spaces. As a drawback, however, Schoenberg means, seem less sensitive for dimensionality of configurations considered than intrinsic or Ziezold means. In particular for problems involving small sample sizes $n$ and a large number of parameters $p$ as currently of high interest in statistcal applications, involving (nearly) degenerate data, Ziezold means may also be preferred over Schoenberg means due to higher power of tests. 
%

	Finally, note that Ziezold means may be defined for the shape spaces of planar curves introduced by \cite{ZR72}, which are currently of interest e.g. \cite{KSMJ04} or \cite{SCC06}. Employing Ziezold means there, a computational advantage greater than found here can be expected since the computation of iterates of intrinsic means involves computations of geodesics which themselves can only be found iteratively. 
\section*{Acknowledgment}
	The author would like to thank 
	Alexander Lytchak for helpful advice 
	on differential geometric issues. 

\appendix
\section{Proofs}


 	\begin{Lem}\label{hor_measurable_lift:lem} Let $U\subset M$ be a tubular neighborhood about $p\in M$ that admits a slice via $\exp_p D \times_{I_p} G \cong U$ in optimal position to $p$. Then, there is a measurable horizontal lift $L\subset \exp_pD$ of $\pi(U)$ in optimal position to $p$.
	\end{Lem}

	\begin{proof}
	 If $p$ is regular, then $L=\exp_pD $ has the desired properties. 
	Now assume that $p$ is not of maximal orbit type.  W.l.o.g. assume that $D$ contains the closed ball $B$ of radius $r>0$ with bounding sphere $S = \partial B$ and that there are $p^1,\ldots, p^J \in \exp_p(S)$ having the distinct orbit types orccuring in $S$. $S^{j}$ denotes all points on $S$ of orbit type $(G/I_{p^j})$, $j=1,\ldots,J$, respectively. Observe that each $S^{j}$ is a manifold on which $I_{p}$ acts isometrically. Hence for every $1\leq j\leq J$, there is a finite ($K_j<\infty$) or countable ($K_j=\infty$) sequence of tubular neighborhoods $U^j_k\subset S^{j}$ covering  $S^{j}$, admitting trivial slices via 
	$$\exp^{S^j}_{p^j_k}D^j_k \times I_p/I_{p_k^j} \cong U^j_k, ~~p_k^j\in U_k^j,~~1\leq k \leq K_j\,.$$
	Here, $\exp^{S^j}_{p^j_k}$ denotes the Riemannian exponential of $S^j$. 
	Defining a disjoint sequence 
	\begin{eqnarray*}
	\widetilde{U}^{j}_1&:=& U_1^j\,,\quad
	\widetilde{U}^{j}_{k+1}~:=~ U_{k+1}^j\setminus \widetilde{U}^{j}_{k}\mbox{ for $1\leq k\leq K_j-1 $}
	\end{eqnarray*}
	exhausting $S^j$ we obtain a corresponding sequence of disjoint measurable sets $\exp^{S^j}_{p^j_k}\widetilde{D}^{j}_{k}$ with  
	$$\exp^{S^j}_{p^j_k}\widetilde{D}^j_k \times I_p/I_{p^j_k} \cong \widetilde{U}^j_k, ~~1\leq k \leq K_j\,.$$
	Hence, setting 
	$$L^{j}_{k} := \exp^{S^j}_{p^j_k}\widetilde{D}^j_k\mbox{ and } L^j := \bigcup_{k=1}^{K_j}L_k^j$$
	 observe that every $p' \in S^j$ has a unique lift in $L^j$ which is contained in a unique $L^{j}_{k}$. This lift is by construction (all $L_k^j$ are in $\exp_pD$) in optimal position to $p$. Moreover, if $p'\in L^j_k$ and $gp' \in L^{j}_{k'}$ for some $g\in G$ with $1\leq k',k\leq K_j$ we have by the disjoint construction of $\widetilde{U}^{j}_{k}$ and $\widetilde{U}^{j}_{k'}$ that $k=k'$, hence the isotropy groups of $gp'$ and $p'$ agree, yielding $gp'=p'$. In consequence, $L^j$ is a measurable horizontal lift of $S^j$ in optimal position to $p$.
	Since every horizontal geodesic segment $t\mapsto \exp_p(tv)$, $v\in H_pM$ contained in $\exp_pD$ features a constant isotropy group, except possibly for the initial point we obtain with the definition of
	$$L :=	\cup_{j=1}^JM^j\mbox{ with }M^j :=\{\exp_p(tv)\in \exp_p(D): v\in \exp_p^{-1}(L^j), t\geq 0\}$$	 
%
	a measurable horizontal lift of $\pi(U)$ in optimal position to $p$.
	\end{proof} 

	\paragraph{Proof of Theorem \ref{glob_hor_measurable_lift:thm}.}

	Since $M$ is connected, any two points $p,p'$ can be brought into optimal position $p,gp'$ and a closed minimizing horizontal geodesic segment $\gamma_{gp'}$  between $p,gp'$ can be found. If $[p']\in Q^{([p])}$ then also $\gamma_{gp'}\subset  Q^{([p])}$. In consequence, there are tubulars neighborhoods $U_{p}$ of $p$ and $U_{p'}$ of $\gamma_{gp'}$ admitting slices in optimal position to $p$, which by Lemma  \ref{hor_measurable_lift:lem}, have horizontal lifts $L_{p}$ and $L_{p'}$ in optimal position to $p$. Since $M$ is a manifold, there is a sequence $[p_0],\ldots \in  Q^{([p])}$, $g_j\in G, p_j \in M$ such that $p_0=p$ and that each $g_jp_j$ is in optimal position to $p$ ($j\in J,~J \subset\mathbb N$) and such that
	$$Q^{([p])}~\subset~ \bigcup_{j\in J\cup\{0\}} \pi(U_{p_j})$$
	with  measurable horizontal lifts $L_{p_j}$ of $\pi(U_{p_j})$.  Defining $L'_{p_0}:=L_{p_0}$ and recursively $L'_{p_{j+1}} := L_{p_{j+1}}\setminus L'_{p_{j}}$ for $j=1,\ldots$ a measurable horizontal lift
	$ L' := \cup_{j=0}^\infty L'_{p_j}$ of $Q^{([p])}$ in optimal position to $p$ is obtained. Finally, suppose that $p_j$ is in optimal position to $p$ for $p_j \in [p_j]\in A$ and set
	$L''_0 := L'$, $L''_j := L''_{j-1}\cup\{p_j\}$ if $[p_j]\cap L'_j = \emptyset$ and $L''_j = L''_{j-1}$ ($j\geq 1$) otherwise to obtain the desired measurable horizontal lift $L'':= \cup_{[p_j]\in A} L''_j$ in optimal position to $p$. 


	\paragraph{Proof of Theorem \ref{population_mean_iso_grp:thm}.} 

	In case of intrinsic means, with the hypotheses and notations of the above proof of Theorem \ref{glob_hor_measurable_lift:thm}, suppose that $L''$ is a measurable  horizontal lift of $Q^{([p])}\cup A$ 
	in optimal position to an intrinsic mean $p\in M$ of the random element $Y$ on $M$ defined as in Theorem \ref{pop_int_mean_bottom_int_mean_top:Th} with $[p]\in  E^{(d_Q)}(X)$. For notational simplicity we assume that $Q^{([p])}=\pi(U)$ with a single tubular neighborhood $U$ of $p$ admitting a slice. 

%
%
 	Then, additionally using the notation of the above proof of Lemma \ref{hor_measurable_lift:lem}, if the assertion of the Theorem would be false, w.l.o.g. there would be $g\in I_p$, $1\leq {j}\leq J$, $p_{j} \in S^{j}$ with $gp_{j} \neq p_{j} $ and $\mathbb P\{Y\in M^{j}\}>0$.
	In particular, in the proof Lemma \ref{hor_measurable_lift:lem}, we may choose a sufficiently small $U^{j}_{k}$ around $p_{j}$ such that in consequence of (\ref{intr-inj:cond})   
	\begin{eqnarray}\label{diff-exp:eq} \int_{M^{j}_{k}(\epsilon)} \big(\exp^{-1}_pY - \exp^{-1}_p (gY)\big) \,d\Prob_Y&\neq& 0\,\end{eqnarray}
	 with some $\epsilon, r>0$,
	$M_k^j(\epsilon) :=\{\exp_p(tv)\in \exp_p(D): v\in \exp_p^{-1}(L^j_k), |t-r|<\epsilon\}$ and $L_k^j$ obtained from $U_k^j$ as in the proof of Lemma \ref{hor_measurable_lift:lem}.
	Suppose that $L\subset L''$ is obtained as in the proof of Lemma \ref{hor_measurable_lift:lem} by using $L^{j}_{k}$ and suppose that $\widetilde{L''}$ is obtained from $L''$ by replacing the $M_k^j(\epsilon)$ part of $M_k^j$ with $\{\exp_p(tv)\in \exp_p(D): v\in \exp_p^{-1}(gL^j_k), |t-r|<\epsilon\}$. 
	Then $\widetilde{L''}$ is also a measurable horizontal lift in optimal position to $p$. Since we assume that $[p]$ is an intrinsic mean of $X$,  assertion (i) of Theorem \ref{pop_int_mean_bottom_int_mean_top:Th} teaches that $p$ is also an intrinsic mean of lift $\widetilde{Y}$ of $X$ to $\widetilde{L''}$, i.e. 
	\begin{eqnarray*}
	 0&=&\int_{L''} \exp^{-1}_pY\,d\Prob_Y - \int_{\widetilde{L''}} \exp^{-1}_p\widetilde{Y}\,d\Prob_Y'\\
	&=&\int_{M^{j}_{k}(\epsilon)} \big(\exp^{-1}_pY - \exp^{-1}_p (gY)\big) \,d\Prob_Y\,.
	\end{eqnarray*}
 	This is a contradiction	to (\ref{diff-exp:eq}) yielding the validity of the theorem for intrinsic means. 

	The assertion in case of Ziezold means is similarly obtained. Use the same horizontal lifts $L''$ and $\widetilde{L''}$ from above, replace $\exp^{-1}_pY$, $\exp^{-1}_p\widetilde{Y}$ and $\exp^{-1}_p (gY)$ by $ d f^Y_{ext}(p)$, $  df^{\widetilde{Y}}_{ext}(p)$ and $ d f_{ext}^{gY}(p)$, respectively, use the hypothesis (\ref{Ziez-inj:cond}) to obtain the analog of (\ref{diff-exp:eq}) and finally obtain the contradiction arguing with assertion (ii) of Theorem \ref{pop_int_mean_bottom_int_mean_top:Th}.

\bibliographystyle{../../BIB/elsart-harv}
\bibliography{../../BIB/shape}

\begin{thebibliography}{40}
\expandafter\ifx\csname natexlab\endcsname\relax\def\natexlab#1{#1}\fi
\expandafter\ifx\csname url\endcsname\relax
  \def\url#1{\texttt{#1}}\fi
\expandafter\ifx\csname urlprefix\endcsname\relax\def\urlprefix{URL }\fi

\bibitem[{Afsari(2010)}]{Afsari10}
Afsari, B., 2010. Riemannian ${L}^p$ center of mass: existence, uniqueness, and
  convexity. Proceedings of the American Mathematical Society 139, 655--773.

\bibitem[{Bandulasiri and Patrangenaru(2005)}]{BandPat05}
Bandulasiri, A., Patrangenaru, V., 2005. Algorithms for nonparametric inference
  on shape manifolds. Proc. of JSM 2005 Minneapolis, MN, 1617--1622.

\bibitem[{Bhattacharya(2008)}]{B08}
Bhattacharya, A., 2008. Statistical analysis on manifolds: A nonparametric
  approach for inference on shape spaces. {S}ankhya, Ser. A 70~(2), 223--266.

\bibitem[{Bhattacharya and Patrangenaru(2003)}]{BP03}
Bhattacharya, R.~N., Patrangenaru, V., 2003. Large sample theory of intrinsic
  and extrinsic sample means on manifolds {I}. Ann. Statist. 31~(1), 1--29.

\bibitem[{Bhattacharya and Patrangenaru(2005)}]{BP05}
Bhattacharya, R.~N., Patrangenaru, V., 2005. Large sample theory of intrinsic
  and extrinsic sample means on manifolds {II}. Ann. Statist. 33~(3),
  1225--1259.

\bibitem[{Bredon(1972)}]{Bre72}
Bredon, G.~E., 1972. Introduction to Compact Transformation Groups. Vol.~46 of
  Pure and Applied Mathematics. Academic Press.

\bibitem[{Choquet(1954)}]{Choq54}
Choquet, G., 1954. {Theory of capacities.} Ann. Inst. Fourier 5, 131--295.

\bibitem[{Dryden et~al.(2008)Dryden, Kume, Le, and Wood}]{DKLW08}
Dryden, I.~L., Kume, A., Le, H., Wood, A. T.~A., 2008. A multidimensional
  scaling approach to shape analysis. To appear.

\bibitem[{Dryden and Mardia(1998)}]{DM98}
Dryden, I.~L., Mardia, K.~V., 1998. Statistical Shape Analysis. Wiley,
  Chichester.

\bibitem[{Fr\'echet(1948)}]{F48}
Fr\'echet, M., 1948. Les \'el\'ements al\'eatoires de nature quelconque dans un
  espace distanci\'e. Ann. Inst. H. Poincar\'e 10~(4), 215--310.

\bibitem[{Gower(1975)}]{Gow}
Gower, J.~C., 1975. Generalized {P}rocrustes analysis. Psychometrika 40,
  33--51.

\bibitem[{Hendriks and Landsman(1996)}]{HL96}
Hendriks, H., Landsman, Z., 1996. Asymptotic behaviour of sample mean location
  for manifolds. Statistics \& Probability Letters 26, 169--178.

\bibitem[{Hendriks and Landsman(1998)}]{HL98}
Hendriks, H., Landsman, Z., 1998. Mean location and sample mean location on
  manifolds: asymptotics, tests, confidence regions. Journal of Multivariate
  Analysis 67, 227--243.

\bibitem[{Hendriks et~al.(1996)Hendriks, Landsman, and Ruymgaart}]{HLR96}
Hendriks, H., Landsman, Z., Ruymgaart, F., 1996. Asymptotic behaviour of sample
  mean direction for spheres. Journal of Multivariate Analysis 59, 141--152.

\bibitem[{Huckemann(2010{\natexlab{a}})}]{H_Procrustes_10}
Huckemann, S., 2010{\natexlab{a}}. Inference on 3{D} {P}rocrustes means: Tree
  boles growth, rank-deficient diffusion tensors and perturbation models.
  Scand. J. Statist., to appear.

\bibitem[{Huckemann(2010{\natexlab{b}})}]{Hshapes}
Huckemann, S., 2010{\natexlab{b}}. R-package for intrinsic statistical analysis
  of shapes,\\
  http:$/\!/$www.mathematik.uni-kassel.de/$\sim$huckeman/software/ishapes\_1.0%
.tar.gz.

\bibitem[{Huckemann et~al.(2010{\natexlab{a}})Huckemann, Hotz, and
  Munk}]{HHM09}
Huckemann, S., Hotz, T., Munk, A., 2010{\natexlab{a}}. Intrinsic {MANOVA} for
  {R}iemannian manifolds with an application to {K}endall's space of planar
  shapes. IEEE Transactions on Pattern Analysis and Machine Intelligence
  32~(4), 593--603.

\bibitem[{Huckemann et~al.(2010{\natexlab{b}})Huckemann, Hotz, and
  Munk}]{HHM07}
Huckemann, S., Hotz, T., Munk, A., 2010{\natexlab{b}}. Intrinsic shape
  analysis: Geodesic principal component analysis for {R}iemannian manifolds
  modulo {L}ie group actions (with discussion). Statistica Sinica 20~(1),
  1--100.

\bibitem[{Huckemann and Ziezold(2006)}]{HZ06}
Huckemann, S., Ziezold, H., 2006. Principal component analysis for {R}iemannian
  manifolds with an application to triangular shape spaces. Adv. Appl. Prob.
  (SGSA) 38~(2), 299--319.

\bibitem[{Jupp(1988)}]{J88}
Jupp, P.~E., 1988. Residuals for directional data. J. Appl. Statist. 15~(2),
  137--147.

\bibitem[{Karcher(1977)}]{Ka77}
Karcher, H., 1977. Riemannian center of mass and mollifier smoothing.
  Communications on Pure and Applied Mathematics XXX, 509--541.

\bibitem[{Kendall(1974)}]{Kend74}
Kendall, D., 1974. {Foundations of a theory of random sets.} {Stochastic Geom.,
  Tribute Memory Rollo Davidson, 322-376 (1974).}

\bibitem[{Kendall et~al.(1999)Kendall, Barden, Carne, and Le}]{KBCL99}
Kendall, D.~G., Barden, D., Carne, T.~K., Le, H., 1999. Shape and Shape Theory.
  Wiley, Chichester.

\bibitem[{Kendall(1990)}]{KWS90}
Kendall, W.~S., 1990. Probability, convexity, and harmonic maps with small
  image {I}: Uniqueness and fine existence. Proc. London Math. Soc. 61,
  371--406.

\bibitem[{Klassen et~al.(2004)Klassen, Srivastava, Mio, and Joshi}]{KSMJ04}
Klassen, E., Srivastava, A., Mio, W., Joshi, S., Mar. 2004. Analysis on planar
  shapes using geodesic paths on shape spaces. IEEE Transactions on Pattern
  Analysis and Machine Intelligence 26~(3), 372--383.

\bibitem[{Kobayashi and Nomizu(1963)}]{KN63}
Kobayashi, S., Nomizu, K., 1963. Foundations of Differential Geometry. Vol.~I.
  Wiley, Chichester.

\bibitem[{Kobayashi and Nomizu(1969)}]{KN69}
Kobayashi, S., Nomizu, K., 1969. Foundations of Differential Geometry. Vol.~II.
  Wiley, Chichester.

\bibitem[{Krim and Yezzi(2006)}]{KY06}
Krim, H., Yezzi, A. J. J.~E., 2006. Statistics and Analysis of Shapes. Modeling
  and Simulation in Science, Engineering and Technology. Birkh\"auser, Boston.

\bibitem[{Le(2001)}]{L01}
Le, H., 2001. Locating {F}r\'{e}chet means with an application to shape spaces.
  Adv. Appl. Prob. (SGSA) 33~(2), 324--338.

\bibitem[{Le(2004)}]{L04}
Le, H., 2004. Estimation of {R}iemannian barycenters. LMS Journal of
  Computation and Mathematics 7, 193--200.

\bibitem[{Mardia and Patrangenaru(2001)}]{MP01}
Mardia, K., Patrangenaru, V., 2001. On affine and projective shape data
  analysis. Functional and Spatial Data Analysis, Proceedings of the 20th LASR
  Workshop (Eds: K.V. Mardia and R.G. Aykroyd), 39--45.

\bibitem[{Mardia and Patrangenaru(2005)}]{MP05}
Mardia, K., Patrangenaru, V., 2005. Directions and projective shapes. Ann.
  Statist. 33, 1666--1699.

\bibitem[{Matheron(1975)}]{Math75}
Matheron, G., 1975. {Random sets and integral geometry.} {Wiley Series in
  Probability and Mathematical Statistics. New York}.

\bibitem[{Nash(1956)}]{Na56}
Nash, J., 1956. The imbedding problem for {R}iemannian manifolds. Ann. of Math.
  63, 20--63.

\bibitem[{Palais(1961)}]{P61}
Palais, R.~S., 1961. On the existence of slices for actions of non-compact
  {L}ie groups. Ann. Math. 2nd Ser. 73~(2), 295--323.

\bibitem[{Schmidt et~al.(2006)Schmidt, Clausen, and Cremers}]{SCC06}
Schmidt, F.~R., Clausen, M., Cremers, D., 2006. Shape matching by variational
  computation of geodesics on a manifold. In: Pattern Recognition (Proc. DAGM).
  Vol. 4174 of LNCS. Springer, Berlin, Germany, pp. 142--151.

\bibitem[{Small(1996)}]{S96}
Small, C.~G., 1996. The Statistical Theory of Shape. Springer-Verlag, New York.

\bibitem[{Zahn and Roskies(1972)}]{ZR72}
Zahn, C., Roskies, R., 1972. Fourier descriptors for plane closed curves. IEEE
  Trans. Computers C-21, 269--281.

\bibitem[{Ziezold(1977)}]{Z77}
Ziezold, H., 1977. Expected figures and a strong law of large numbers for
  random elements in quasi-metric spaces. Trans. 7th Prague Conf. Inf. Theory,
  Stat. Dec. Func., Random Processes A, 591--602.

\bibitem[{Ziezold(1994)}]{Z94}
Ziezold, H., 1994. Mean figures and mean shapes applied to biological figure
  and shape distributions in the plane. Biom. J.~(36), 491--510.

\end{thebibliography}


\begin{thebibliography}{42}
\expandafter\ifx\csname natexlab\endcsname\relax\def\natexlab#1{#1}\fi
\expandafter\ifx\csname url\endcsname\relax
  \def\url#1{\texttt{#1}}\fi
\expandafter\ifx\csname urlprefix\endcsname\relax\def\urlprefix{URL }\fi

\bibitem[{Bandulasiri and Patrangenaru(2005)}]{BandPat05}
Bandulasiri, A., Patrangenaru, V., 2005. Algorithms for nonparametric inference
  on shape manifolds.

\bibitem[{Bhattacharya(2008)}]{B08}
Bhattacharya, A., 2008. Statistical analysis on manifolds: A nonparametric
  approach for inference on shape spaces. ,~to appear.

\bibitem[{Bhattacharya and Patrangenaru(2003)}]{BP03}
Bhattacharya, R.~N., Patrangenaru, V., 2003. Large sample theory of intrinsic
  and extrinsic sample means on manifolds {I}. Ann. Statist. 31~(1), 1--29.

\bibitem[{Bhattacharya and Patrangenaru(2005)}]{BP05}
Bhattacharya, R.~N., Patrangenaru, V., 2005. Large sample theory of intrinsic
  and extrinsic sample means on manifolds {II}. Ann. Statist. 33~(3),
  1225--1259.

\bibitem[{Blum and Nagel(1978)}]{BN78}
Blum, H., Nagel, R.~N., 1978. Shape description using weighted symmetric axis
  features. Pattern Recognition 10~(3), 167--180.

\bibitem[{Bredon(1972)}]{Bre72}
Bredon, G.~E., 1972. Introduction to Compact Transformation Groups. Vol.~46 of
  Pure and Applied Mathematics. Academic Press.

\bibitem[{Choquet(1954)}]{Choq54}
Choquet, G., 1954. {Theory of capacities.} Ann. Inst. Fourier 5, 131--295.

\bibitem[{Dryden et~al.(2008)Dryden, Kume, Le, and Wood}]{DKLW08}
Dryden, I.~L., Kume, A., Le, H., Wood, A. T.~A., 2008. A multidimensional
  scaling approach to shape analysis. To appear.

\bibitem[{Dryden and Mardia(1998)}]{DM98}
Dryden, I.~L., Mardia, K.~V., 1998. Statistical Shape Analysis. Wiley,
  Chichester.

\bibitem[{Fr\'echet(1948)}]{F48}
Fr\'echet, M., 1948. Les \'el\'ements al\'eatoires de nature quelconque dans un
  espace distanci\'e. Ann. Inst. H. Poincar\'e 10~(4), 215--310.

\bibitem[{Gower(1975)}]{Gow}
Gower, J.~C., 1975. Generalized {P}rocrustes analysis. Psychometrika 40,
  33--51.

\bibitem[{Hendriks(1991)}]{H91}
Hendriks, H., 1991. A {C}ram\'er-{R}ao type lower bound for estimators with
  values in a manifold. Journal of Multivariate Analysis 38, 245--261.

\bibitem[{Hendriks and Landsman(1996)}]{HL96}
Hendriks, H., Landsman, Z., 1996. Asymptotic behaviour of sample mean location
  for manifolds. Statistics \& Probability Letters 26, 169--178.

\bibitem[{Hendriks and Landsman(1998)}]{HL98}
Hendriks, H., Landsman, Z., 1998. Mean location and sample mean location on
  manifolds: asymptotics, tests, confidence regions. Journal of Multivariate
  Analysis 67, 227--243.

\bibitem[{Hendriks et~al.(1996)Hendriks, Landsman, and Ruymgaart}]{HLR96}
Hendriks, H., Landsman, Z., Ruymgaart, F., 1996. Asymptotic behaviour of sample
  mean direction for spheres. Journal of Multivariate Analysis 59, 141--152.

\bibitem[{Hobolth et~al.(2002)Hobolth, Kent, and Dryden}]{HKD02}
Hobolth, A., Kent, J., Dryden, I., 2002. On the relation between edge and
  vertex modelling in shape analysis. Scandinavian Journal of Statistics
  29~(3), 355--374.

\bibitem[{Hotz et~al.(2010)Hotz, Huckemann, Gaffrey, Munk, and
  Sloboda}]{HHGMS07}
Hotz, T., Huckemann, S., Gaffrey, D., Munk, A., Sloboda, B., 2010. Shape spaces
  for pre-alingend star-shaped objects in studying the growth of plants.
  Journal of the Royal Statistical Society, Series C 59~(1), 127--143.

\bibitem[{Huckemann et~al.(2009)Huckemann, Hotz, and Munk}]{HHM09}
Huckemann, S., Hotz, T., Munk, A., 2009. Intrinsic {MANOVA} for {R}iemannian
  manifolds with an application to {K}endall's space of planar shapes. IEEE
  Transactions on Pattern Analysis and Machine Intelligence.~To appear.

\bibitem[{Huckemann and Ziezold(2006)}]{HZ06}
Huckemann, S., Ziezold, H., 2006. Principal component analysis for {R}iemannian
  manifolds with an application to triangular shape spaces. Adv. Appl. Prob.
  (SGSA) 38~(2), 299--319.

\bibitem[{Jupp(1988)}]{J88}
Jupp, P.~E., 1988. Residuals for directional data. J. Appl. Statist. 15~(2),
  137--147.

\bibitem[{Karcher(1977)}]{Ka77}
Karcher, H., 1977. Riemannian center of mass and mollifier smoothing.
  Communications on Pure and Applied Mathematics XXX, 509--541.

\bibitem[{Kendall(1974)}]{Kend74}
Kendall, D., 1974. {Foundations of a theory of random sets.} {Stochastic Geom.,
  Tribute Memory Rollo Davidson, 322-376 (1974).}

\bibitem[{Kendall(1984)}]{K84}
Kendall, D.~G., 1984. Shape manifolds, {P}rocrustean metrics and complex
  projective spaces. Bull. Lond. Math. Soc. 16~(2), 81--121.

\bibitem[{Kendall et~al.(1999)Kendall, Barden, Carne, and Le}]{KBCL99}
Kendall, D.~G., Barden, D., Carne, T.~K., Le, H., 1999. Shape and Shape Theory.
  Wiley, Chichester.

\bibitem[{Kendall(1990)}]{KWS90}
Kendall, W.~S., 1990. Probability, convexity, and harmonic maps with small
  image {I}: Uniqueness and fine existence. Proc. London Math. Soc. 61,
  371--406.

\bibitem[{Klassen et~al.(2004)Klassen, Srivastava, Mio, and Joshi}]{KSMJ04}
Klassen, E., Srivastava, A., Mio, W., Joshi, S., Mar. 2004. Analysis on planar
  shapes using geodesic paths on shape spaces. IEEE Transactions on Pattern
  Analysis and Machine Intelligence 26~(3), 372--383.

\bibitem[{Kobayashi and Nomizu(1969)}]{KN69}
Kobayashi, S., Nomizu, K., 1969. Foundations of Differential Geometry. Vol.~II.
  Wiley, Chichester.

\bibitem[{Krim and Yezzi(2006)}]{KY06}
Krim, H., Yezzi, A. J. J.~E., 2006. Statistics and Analysis of Shapes. Modeling
  and Simulation in Science, Engineering and Technology. Birkh\"auser, Boston.

\bibitem[{Le(2001)}]{L01}
Le, H., Jun. 2001. Locating {F}r\'{e}chet means with an application to shape
  spaces. Adv. Appl. Prob. (SGSA) 33~(2), 324--338.

\bibitem[{Le(2004)}]{L04}
Le, H., 2004. Estimation of {R}iemannian barycenters. LMS Journal of
  Computation and Mathematics 7, 193--200.

\bibitem[{Mardia and Patrangenaru(2001)}]{MP01}
Mardia, K., Patrangenaru, V., 2001. On affine and projective shape data
  analysis. Functional and Spatial Data Analysis, Proceedings of the 20th LASR
  Workshop (Eds: K.V. Mardia and R.G. Aykroyd), 39--45.

\bibitem[{Mardia and Patrangenaru(2005)}]{MP05}
Mardia, K., Patrangenaru, V., 2005. Directions and projective shapes. Ann.
  Statist. 33, 1666--1699.

\bibitem[{Matheron(1975)}]{Math75}
Matheron, G., 1975. {Random sets and integral geometry.} {Wiley Series in
  Probability and Mathematical Statistics. New York}.

\bibitem[{Mosimann(1970)}]{M70}
Mosimann, J.~E., 1970. Size allometry: Size and shape variables with
  characterizations of the lognormal and generalized gamma distributions.
  Journal of the American Statistical Association 65~(330), 930--945.

\bibitem[{Munk et~al.(2008)Munk, Paige, Pang, Patrangenaru, and
  Ruymgaart}]{MPPPR07}
Munk, A., Paige, R., Pang, J., Patrangenaru, V., Ruymgaart, F., 2008. The one-
  and multi-sample problem for functional data with application to projective
  shape analysis. Journal of Multivariate Analysis~(99), 815--833.

\bibitem[{Nash(1956)}]{Na56}
Nash, J., 1956. The imbedding problem for {R}iemannian manifolds. Ann. of Math.
  63, 20--63.

\bibitem[{Palais(1960)}]{P60}
Palais, R.~S., 1960. Slices and equivariant embeddings. In: Borel, A. (Ed.),
  Seminar on Transformation Groups. No.~46 in Annals of Mathematics Studies.
  Princeton Univ. Press, Princton NJ, Ch.~8, pp. 101--115.

\bibitem[{Schmidt et~al.(2006)Schmidt, Clausen, and Cremers}]{SCC06}
Schmidt, F.~R., Clausen, M., Cremers, D., 2006. Shape matching by variational
  computation of geodesics on a manifold. In: Pattern Recognition (Proc. DAGM).
  Vol. 4174 of LNCS. Springer, Berlin, Germany, pp. 142--151.

\bibitem[{Small(1996)}]{S96}
Small, C.~G., 1996. The Statistical Theory of Shape. Springer-Verlag, New York.

\bibitem[{Zahn and Roskies(1972)}]{ZR72}
Zahn, C., Roskies, R., 1972. Fourier descriptors for plane closed curves. IEEE
  Trans. Computers C-21, 269--281.

\bibitem[{Ziezold(1977)}]{Z77}
Ziezold, H., 1977. Expected figures and a strong law of large numbers for
  random elements in quasi-metric spaces. Trans. 7th Prague Conf. Inf. Theory,
  Stat. Dec. Func., Random Processes A, 591--602.

\bibitem[{Ziezold(1994)}]{Z94}
Ziezold, H., 1994. Mean figures and mean shapes applied to biological figure
  and shape distributions in the plane. Biom. J.~(36), 491--510.

\end{thebibliography}

\end{document}